\pdfoutput=1

\newif\ifdraft\draftfalse   
\newif\ifanon\anonfalse     
\newif\ifcamera\cameratrue  
\newif\ifshepherding\shepherdingfalse 
\newif\ifappendix\appendixfalse 
\newif\iffull\fullfalse     
\newif\iflongrefs\longrefsfalse 
\newif\ifbackref\backreffalse 
\newif\ifsooner\soonerfalse
\newif\iflater\laterfalse
\newif\ifprob\probfalse     
\makeatletter \@input{texdirectives.tex} \makeatother

\ifcamera
  \ifshepherding
  \documentclass[acmsmall,review,nonacm]{acmart}
  \else
  \documentclass[acmsmall, screen]{acmart}
  \fi
\else
  \ifanon
  \documentclass[acmsmall,review,anonymous,nonacm]{acmart}
  \else
  \documentclass[acmsmall,review=false,screen=true]{acmart}
  \fi
\fi

\acmPrice{}
\acmDOI{10.1145/3371072}
\acmYear{2020}
\copyrightyear{2020}
\acmJournal{PACMPL}
\acmVolume{4}
\acmNumber{POPL}
\acmArticle{4}
\acmMonth{1}

\ifanon
\setcopyright{none}             
\else
\setcopyright{rightsretained}   
\fi

\makeatletter
\def\@copyrightpermission{\ifcamera\\\\\\\fi This work is licensed under a \href{https://creativecommons.org/licenses/by/4.0/}{Creative Commons Attribution 4.0 International License}}
\makeatother

\ifcamera
  \ifshepherding
  \makeatletter
  \def\@authorsaddresses{}
  \makeatother
  \fi
\else
\makeatletter
\def\@authorsaddresses{}
\fancypagestyle{firstpagestyle}{%
  \fancyhf{}%
  \renewcommand{\headrulewidth}{\z@}%
  \renewcommand{\footrulewidth}{\z@}%
    \fancyhead[L]{\ifanon\ACM@linecountL\fi}%
}
\fancypagestyle{standardpagestyle}{%
  \fancyhf{}%
  \renewcommand{\headrulewidth}{\z@}%
  \renewcommand{\footrulewidth}{\z@}%
    \fancyhead[LE]{\ifanon\ACM@linecountL\fi\@headfootfont\thepage}%
    \fancyhead[RO]{\@headfootfont\thepage}%
    \fancyhead[RE]{\@headfootfont\@shortauthors}%
    \fancyhead[LO]{\ifanon\ACM@linecountL\fi\@headfootfont\shorttitle}%
    \fancyfoot[RO,LE]{}%
}
\def\@mkbibcitation{}

\makeatother
\fi

\usepackage{xstring}
\usepackage{afterpage}
\usepackage{color}
\usepackage{amssymb,amsmath,amsfonts}
\usepackage{xspace}
\usepackage{graphicx}
\usepackage{listings}
\usepackage{mathpartir} 
\usepackage{natbib}
\newcommand\citepos[1]{\citeauthor{#1}'s\ [\citeyear{#1}]}

\let\cite=\citep
\usepackage{amsthm}
\usepackage{stmaryrd}
\usepackage{bbold}
\usepackage{enumitem}
\usepackage{tikz}
\usepackage{tikz-cd}
\usepackage[T1]{fontenc} 
\usepackage[utf8]{inputenc}
\usepackage[all]{xy}

\DeclareUnicodeCharacter{00A0}{~}

\usepackage{combelow}

\ifcamera
\usepackage{flushend} 
\fi
\definecolor{darkblue}{rgb}{0.0,0.0,0.3}

\usepackage{hyperref}
\hypersetup{breaklinks=true}
\usepackage{soul} 
\newcommand\maybecolor[1]{\color{#1}}

\let\ls\lstinline

\def\Snospace~{\S{}}

\newcommand\fstar{F$^\star$\xspace}
\newcommand{\coq}{\textsc{Coq}}

\newcommand{\RDTT}{\textsc{RDTT}}
\newcommand\RHL{\textsc{RHL}}
\newcommand{\RHTT}[0]{\textsc{RHTT}}
\newcommand{\While}{\textsc{While}}

\definecolor{dkblue}{rgb}{0,0.1,0.5}
\definecolor{dkgreen}{rgb}{0,0.4,0}
\definecolor{dkred}{rgb}{0.6,0,0}
\definecolor{dkpurple}{rgb}{0.7,0,1.0}
\definecolor{purple}{rgb}{0.9,0,1.0}
\definecolor{olive}{rgb}{0.4, 0.4, 0.0}
\definecolor{teal}{rgb}{0.0,0.4,0.4}
\definecolor{azure}{rgb}{0.0, 0.5, 1.0}
\definecolor{gray}{rgb}{0.5, 0.5, 0.5}
\definecolor{dkgrey}{rgb}{0.2, 0.2, 0.2}
\definecolor{lilac}{rgb}{0.70, 0.04, 0.08}

\newcommand{\comm}[3]{\ifdraft{{\color{#1}[#2: #3]}}\fi}

\newcommand{\ch}[1]{\comm{teal}{CH}{#1}}

\newcommand{\km}[1]{\comm{dkred}{Kenji}{#1}}

\newcommand{\er}[1]{\comm{dkgreen}{ER}{#1}}

\newcommand*{\EG}{e.g.,\xspace}
\newcommand*{\IE}{i.e.,\xspace}

\newcommand*{\ETC}{etc.\xspace}

\newtheorem{definition}{Definition}
\newtheorem{theorem}{Theorem}

\newcommand*{\ii}[1]{\ensuremath{\mathit{#1}}}

\newcommand\alt[0]{\;|\;}

\newcommand\Pure{\mathrm{Pure}}

\newcommand\St{\mathrm{St}}
\newcommand\StT{\mathrm{StT}}
\newcommand\Exc{\mathrm{Exc}}
\newcommand{\Err}{\mathrm{Err}}

\newcommand\ExcT{\mathrm{ExcT}}

\newcommand{\Id}{\mathrm{Id}}
\newcommand{\Part}{\mathrm{Part}}
\newcommand{\Tot}{\mathrm{Tot}}

\newcommand\Imp{\mathrm{Imp}}

\newcommand{\pre}{\mathit{pre}}
\newcommand{\post}{\mathit{post}}

\newcommand\prop{\mathbb{P}}

\newcommand\M{\mathrm{M}}
\newcommand{\Mone}{{\M_1}}
\newcommand{\Mtwo}{{\M_2}}
\newcommand\W{\mathrm{W}}
\newcommand\Wone[1]{{\W^{#1}_1}}
\newcommand\Wtwo[1]{{\W^{#1}_2}}
\renewcommand\P{\mathrm{P}}
\newcommand{\relobj}{\text{rel}}
\newcommand{\relmarker}[0]{\scalebox{0.5}{$\relobj$}}
\newcommand\Wrel[1]{{\W^{#1}_{\!\relmarker}}}
\newcommand{\PP}[0]{\mathrm{PP}}
\newcommand\PPrel[1]{{\PP^{#1}_{\!\relmarker}}}
\newcommand\T{\mathcal{T}}
\newcommand\I{\mathcal{I}}
\renewcommand\C{\mathcal{C}}

\newcommand{\ret}{\texttt{ret}}

\newcommand\retM{\ret^{\M}}

\newcommand{\bind}{\texttt{bind}}

\newcommand\bindM{\bind^{\M}}

\newcommand{\revBind}[2]{{#1}^{\dagger_{#2}}}
\newcommand{\revBindT}[1]{\revBind{#1}{\T}}

\newcommand{\prodrel}[3]{{#1}{\times}{#2}{\leadsto}{#3}}

\newcommand\bool[0]{\mathbb{B}}
\newcommand{\true}[0]{\texttt{tt}}
\newcommand{\false}[0]{\texttt{ff}}
\newcommand{\choice}[0]{\texttt{pick}}
\renewcommand{\choose}[0]{\texttt{choose}}
\newcommand{\fail}[0]{\texttt{fail}}
\newcommand{\NDet}[0]{\mathrm{Nd}}
\newcommand{\IO}[0]{\mathrm{IO}}
\newcommand{\Prob}[0]{\mathrm{Prob}}
\newcommand{\supp}[1]{\mathrm{supp}(#1)}

\newcommand{\inr}{\textcolor{dkred}{\texttt{Inr}}}
\newcommand{\inl}{\textcolor{dkred}{\texttt{Inl}}}

\newcommand{\stget}[0]{\texttt{get}}
\newcommand{\stput}[0]{\texttt{put}}

\newcommand{\Val}{\mathcal{V}\mathrm{al}}

\newcommand{\throw}[0]{\texttt{throw}}
\newcommand{\catch}[0]{\texttt{catch}}

\newcommand{\ioread}[0]{\texttt{input}}
\newcommand{\iowrite}[0]{\texttt{output}}

\newcommand{\flip}{\texttt{flip}}
\newcommand{\interval}{\mathbb{I}}

\newcommand{\fin}{\texttt{fin}}
\newcommand{\dowhile}{\texttt{do\_while}}

\newcommand{\kwd}[1]{\texttt{#1}}

\newcommand{\cskip}{\kwd{skip}}
\newcommand{\cloop}{\kwd{loop}}

\newcommand\abs[2]{\lambda#1.~#2}
\newcommand{\Ifte}[3]{\kwd{if}\> #1 \>\kwd{then}\> #2 \>\kwd{else}\> #3}
\newcommand{\case}[5]{\texttt{case }#1~[#2.#3\alt #4.#5]}
\newcommand{\caseL}[5]{\texttt{case}\,#1\,\left [{\begin{array}{l}#2.#3\\#4.#5\end{array}}\right]}

\newcommand{\Type}[0]{\mathrm{Type}}
\newcommand{\one}[0]{\mathbb{1}}
\newcommand{\zero}[0]{\mathbb{0}}

\newcommand{\br}[1]{{#1}^{\boldsymbol{r}}}
\newcommand{\sem}[1]{\llbracket #1 \rrbracket}
\newcommand{\bN}{\mathbb{N}}

\newcommand{\invariant}{\kwd{inv}}

\newcommand{\srelsyn}[3]{\vdash #1 \sim #2 ~\left\{\>#3\>\right\}}
\newcommand{\srelsem}[4]{\vDash_{#4} #1 \sim #2 ~\left\{\>#3\>\right\}}

\newcommand{\crelsyn}[6]{#1 \vdash #2~\{#3\} \sim #4~\{#5\}~|~#6}
\newcommand{\crelsynL}[6]{#1 {\vdash} \left.{\begin{array}{cc} #2& \{#3\} \\[-0.8ex] \sim& \\[-0.8ex] #4 & \{#5\} \end{array}} \right |#6}
\newcommand{\crelsemext}[7]{#1 \vDash_{#7} #2~\{#3\} \sim #4~\{#5\}~|~#6}
\newcommand{\crelsem}[6]{\crelsemext{#1}{#2}{#3}{#4}{#5}{#6}{}}

\newcommand{\dottimes}{\mathbin{\dot{\times}}}

\newcommand\id{\mathrm{id}}

\newcommand{\J}{\mathcal{J}}

\newcommand{\TypeCat}{\mathcal{T}\!\!\mathit{ype}}
\newcommand{\OrdCat}{\mathcal{O}\hspace{-0.5mm}\mathit{rd}}
\newcommand{\Poset}[0]{\OrdCat}

\newcommand{\Discr}{\mathrm{Disc}}

\newcommand{\Jprod}{{\mathcal{J}_\times}}
\newcommand{\Span}[1]{\mathcal{S}\mathrm{pan}(#1)}

\usepackage{letltxmacro}
\newcommand*{\SavedLstInline}{}
\LetLtxMacro\SavedLstInline\lstinline
\DeclareRobustCommand*{\lstinline}{%
  \ifmmode
    \let\SavedBGroup\bgroup
    \def\bgroup{%
      \let\bgroup\SavedBGroup
      \hbox\bgroup
    }%
  \fi
  \SavedLstInline
}

\allowdisplaybreaks

\renewcommand{\paragraph}[1]{\smallskip{\bf #1.}\;}

\begin{document}


\setlength{\abovedisplayskip}{2pt}
\setlength{\belowdisplayskip}{2pt}

\def\titlestr{The Next 700 Relational Program Logics}
\title[\titlestr]{\titlestr}



\ifanon
\author{}
\else
\author{Kenji Maillard} 
\affiliation{
  \institution{Inria Paris}\ifcamera\city{Paris}
  \fi}
\affiliation{
  \institution{ENS Paris}\ifcamera\city{Paris}\country{France}\fi}
\saveexpandmode
\noexpandarg
\author{Cătălin Hrițcu}
\restoreexpandmode
\affiliation{
  \institution{Inria Paris}\ifcamera\city{Paris}\country{France}\fi}
\author{Exequiel Rivas} 
\affiliation{
  \institution{Inria Paris}\ifcamera\city{Paris}\country{France}\fi}
\author{Antoine Van Muylder} %
\affiliation{
  \institution{Inria Paris}\ifcamera\city{Paris}
  \fi}
\affiliation{
  \ifcamera\institution{Université de Paris}\city{Paris}\country{France}
  \else\institution{Université de Paris}\fi}
\fi


\begin{abstract}
We propose the first framework for defining relational program logics for
arbitrary monadic effects.  The framework is embedded within a relational
dependent type theory and is highly expressive.  At the semantic level, we
provide an algebraic presentation of relational specifications as a class
of relative monads, and link computations and specifications by introducing
relational effect observations, which map pairs of monadic computations to
relational specifications in a way that respects the algebraic structure.  For
an arbitrary relational effect observation, we generically define the core of a
sound relational program logic, and explain how to complete it to a full-fledged
logic for the monadic effect at hand. We show that this generic framework can be
used to define relational program logics for effects as diverse as state,
input-output, nondeterminism, and discrete probabilities.
We, moreover, show that by instantiating our
framework with state and unbounded iteration we can embed a variant of Benton's
Relational Hoare Logic, and also sketch how to reconstruct Relational Hoare Type
Theory.  Finally, we identify and overcome conceptual challenges that prevented
previous relational program logics from properly dealing with control effects,
and are the first to provide a relational program logic for exceptions.

\end{abstract}

\ccsdesc[500]{Theory of computation~Type theory}
\ccsdesc[500]{Theory of computation~Program specifications}
\ccsdesc[500]{Theory of computation~Program verification}
\ccsdesc[300]{Theory of computation~Categorical semantics}
\ccsdesc[300]{Theory of computation~Program reasoning}
\ccsdesc[300]{Theory of computation~Pre- and post-conditions}
\ccsdesc[100]{Theory of computation~Hoare logic}

\keywords{program verification,
relational program logics,
side-effects,
monads,
state,
I/O,
nondeterminism,
probabilities,
exceptions,
dependent types,
semantics,
relative monads,
foundations}

\maketitle

\section{Introduction}
\label{sec:intro}
Generalizing unary properties, which describe single program runs,
{\em relational properties} describe relations between multiple runs
of one or more programs \cite{clarkson10hyp, AbateBGHPT19}.
Formally verifying relational properties has a broad range of practical
applications.
%
%
For instance, one might be interested in proving that the observable behaviors
of two programs are related, showing for instance that the programs are
\emph{equivalent} \cite{CiobacaLRR16, KunduTL09, GodlinS10, BlanchetAF08,
  ChadhaCCK16, TimanySKB18, WangDLC18, Yang07, HurDNV12, HurNDBV14},
or that one \emph{refines} the other \cite{TimanyB19}.
In other cases, one might be interested in relating two runs of a single
program, but,
as soon as the control flow can differ between the two runs, the compositional
verification problem becomes the same as relating two different programs.
This is for instance the case for \emph{noninterference}, which requires that a
program's public outputs are independent of its private
inputs \cite{NanevskiBG13, clarkson10hyp, SabelfeldM03, AntonopoulosGHKTW17,
  BartheEGGKM19, SousaD16, BanerjeeNN16}.
%
The list of practical applications of relational verification is, however, much
longer, including verified program transformations \cite{benton04relational},
cost analysis \cite{CicekBG0H17, QuGG19, RadicekBG0Z18},
program approximation \cite{CarbinKMR12, HeLR18},
semantic diffing \cite{LahiriHKR12, GirkaMR15, GirkaMR17, WangDLC18},
cryptographic proofs \cite{BartheGB09, BartheDGKSS13,
  PetcherM15, BartheFGSSB14, Unruh19},
differential privacy \cite{BartheKOB13, BartheGAHRS15, ZhangK17, Gavazzo18}, and
even machine learning \cite{SatoABGGH19}.




As such, many different relational verification tools have been proposed, making
different tradeoffs, for instance between automation and expressiveness (see
\autoref{sec:related} for further discussion).
In this paper we focus on \emph{relational program logics}, which are a popular
formal foundation for various relational verification tools.
Relational program logics are proof systems whose rules can be used to prove
that a pair of programs meets a rich relational specification.
As such they are very expressive, and can in particular handle situations in
which verifying the desired relational properties requires showing the full
functional correctness of certain pieces of code.
Yet they can often greatly simplify reasoning by leveraging the syntactic
similarities between the programs one relates.
Since \citepos{benton04relational} seminal Relational Hoare Logic, many
relational program logics have been proposed \cite{Yang07, BartheKOB13,
  BartheGAHRS15, PetcherM15, ZhangK17, NanevskiBG13, AguirreBGGS17,
  RadicekBG0Z18, SatoABGGH19, BartheFGSSB14, BartheCK16, QuGG19, SousaD16,
  CarbinKMR12, BanerjeeNN16, Unruh19}.
However, each of these logics is specific to a particular combination of
side-effects that is completely fixed by the programming language and
verification framework; the most popular side-effects these logics bake in
are mutable state, recursion, cost, and probabilities.

The goal of this paper is to distil the generic relational reasoning principles
that work for a broad class of side-effects and that underlie most relational
program logics.
%
%
We do this by introducing the first framework for defining program logics for
{\em arbitrary} monadic effects.
Our generic framework is embedded within a dependent type theory, \EG \coq{},
which makes it highly expressive and simpler to describe.


\paragraph{Syntactic rules}
To factor out the fully generic parts, the rules of the relational program logics
derived in our framework are divided into three categories, following the
syntactic shape of the monadic programs on which they operate:
\begin{enumerate}[label=R\arabic*]
\item\label{rulesetOne} rules for pure language constructs,
  derived from the ambient dependent type theory
  (these rules target the elimination constructs
  for positive types, like if-then-else for booleans,
  recursors for inductive types, \ETC);
\item\label{rulesetTwo} rules for the generic monadic constructs return and bind; and
\item\label{rulesetThree} rules for effect-specific operations (\EG get and put for the state monad,
    or throw and catch for the exception monad).
\end{enumerate}
%
This organization allows us to clearly separate 
not only the generic parts (\ref{rulesetOne}\&\ref{rulesetTwo}) from the effect-specific ones (\ref{rulesetThree}),
but also  the effect-irrelevant parts (\ref{rulesetOne}) from the effect-relevant ones (\ref{rulesetTwo}\&\ref{rulesetThree}).


In its simplest form (\autoref{sec:simplified}), the judgment of the
relational program logics of our framework has the shape:
$\srelsyn{c_1}{c_2}{w}$, where $c_1:\M_1~A_1$ is a computation in
monad $\M_1$ producing results of type $A_1$, where $c_2:M_2~A_2$ is a
computation in monad $\M_2$ producing results of type $A_2$, and where
$w$ is a relational specification of computations $c_1$ and
$c_2$ drawn from the type $\Wrel{}(A_1,A_2)$.
%
Here $\M_1$ and $\M_2$ are two arbitrary and {\em potentially distinct} computation
monads (\EG the state monad $\St~A = S \to A \times S$ and the exception monad
$\Exc~A = A + E$), while $w$ could, for instance, be a pair of a relational
precondition and a relational postcondition, or a relational predicate
transformer---in this introduction we will use relational weakest preconditions.
For instance, for relating two state monads on states
$S_1$ and $S_2$, we often use relational specifications drawn from
\[
\Wrel{\St}(A_1,A_2)  =
  ((A_1 \times S_1) \times (A_2 \times S_2) \to \prop) \to
  S_1 \times S_2 \to \prop
\]
which are predicate transformers mapping postconditions relating two
pairs of a result value and a final state to a precondition relating
two initial states (here $\prop$ stands for the type of propositions
of our \iffull ambient \fi dependent type theory).
As an example of the judgment above, consider
the programs $c_1 =\,$
\lstinline!bindSt \,(get ()) (fun x . put (x + k))!,
which increments the content of a memory cell by $k$, and $c_2
=$ \lstinline!retSt \,()!, which does nothing.
These two programs are related by the specification $w
= \abs{\varphi\,(s_1,s_2)}{\varphi\,(((), s_1+k), ((), s_2))}
: \Wrel{\St}(\one, \one)$ saying that for the
postcondition $\varphi$ to hold for the final states of $c_1$ and $c_2$, it is
enough for it to hold for $s_1+k$ and $s_2$, where $s_1$ and $s_2$
are the computation's initial states.
Note that since $c_1$, $c_2$, and $w$ are terms of our ambient type theory,
free variables (like $k$) are handled directly by the type
theory, which saves the simple judgment from an explicit context.
%

For pure language constructs \ref{rulesetOne}, we try to use the reasoning
principles of our ambient dependent type theory as directly as possible.
For instance, our framework
(again in its simplest incarnation from \autoref{sec:simplified})
provides 
the following rule for the if-then-else construct:
\[
  \inferrule{\Ifte{b}{\srelsyn{c_1}{c_2}{w^\top}}{\srelsyn{c_1}{c_2}{w^\bot}}}
  {\srelsyn{c_1}{c_2}{\Ifte{b}{w^\top}{w^\bot}}}
\]
In order to prove that $c_1$ and $c_2$ satisfy the relational specification
$\Ifte{b}{w_\top}{w_\bot}$, it is enough to prove that  $c_1$ and $c_2$ satisfy
both branches of the conditional in a context extended with the value of $b$.
Interestingly, this rule does not make any assumption on the shape of
\iffull the programs \fi $c_1$ and $c_2$.
%
Relational program logics often classify each rule depending on whether it
considers a syntactic construct that appears on both sides (synchronous), or
only on one side (asynchronous).
In the rule above, taking $c_1$ to be of the shape
$\Ifte{b}{c^\top_1}{c^\bot_1}$ and $c_2$ to be independent of $b$, we can
simplify the premise according to the possible values of $b$
to derive an asynchronous variant of the rule:
\begin{align}
  \label{eq:simple-rule-if-left}
  \inferrule{{\srelsyn{c^\top_1}{c_2}{w^\top}}\\{\srelsyn{c^\bot_1}{c_2}{w^\bot}}}
  {\srelsyn{\Ifte{b}{c^\top_1}{c^\bot_1}}{c_2}{\Ifte{b}{w^\top}{w^\bot}}}
\end{align}
By requiring that both commands are conditionals,
we can also derive the synchronous rule:
\begin{align}
  \label{eq:simple-rule-if-sync}
  \inferrule{\srelsyn{c^\top_1}{c^\top_2}{w^\top}\\ \srelsyn{c^\bot_1}{c^\bot_2}{w^\bot}}
  {\srelsyn{\Ifte{b_1}{c^\top_1}{c^\bot_1}}{\Ifte{b_2}{c^\top_2}{c^\bot_2}}{w^{\bullet}}}
\end{align}
where the relational specification $w^{\bullet} = \abs{\varphi\,s_{12}}{(b
  {\Leftrightarrow} b_1) \wedge (b {\Leftrightarrow} b_2) \wedge
  \Ifte{b}{w^\top\,\varphi\,s_{12}}{w^\bot\,\varphi\,s_{12}}}$ ensures that the
booleans $b_1$ and $b_2$ controlling the choice of the branch in each
computation share the same value $b$.

For the monadic constructs \ref{rulesetTwo}, the challenge is in lifting the
binds and returns of the two computation monads $\Mone$ and $\Mtwo$ to the
specification level.
For instance, in a synchronous rule one would relate $\bind^{\Mone}~m_1~f_1$ to
$\bind^{M_2}~m_2~f_2$ by first relating computations $m_1$ and $m_2$, say via
relational specification $w^m$, and then one would 
relate the two functions $f_1$ and $f_2$
pointwise via a function $w^f$ mapping arguments in $A_1 \times A_2$ to
relational specifications:
\begin{align}
  \label{eq:bindRuleSimple}
  \inferrule{\srelsyn{m_1}{m_2}{w^m}\\
    \forall a_1,a_2
    \srelsyn{f_1\,a_1}{f_2\,a_2}{w^f\,(a_1,a_2)}}
  {\srelsyn{\bind^\Mone\,m_1\,f_1}{\bind^\Mtwo\,m_2\,f_2}{\bind^\Wrel{}\,w^m\,w^f}}
\end{align}
In the conclusion of this rule, we need a way to compose $w : \Wrel{}(A_1,A_2)$ and
$w^f : A_1 \times A_2 \to \Wrel{}(B_1,B_2)$ to obtain a relational specification for the
two binds.
We do this via a bind-like construct:
\begin{align}
  \label{eq:bindWrel-type}
  \bind^\Wrel{} : \Wrel{}(A_1,A_2) \to (A_1 \times A_2 \to \Wrel{}(B_1,B_2)) \to \Wrel{}(B_1,B_2)
\end{align}
For the concrete case of $\Wrel\St$, this bind-like construct takes the form
\[
  \bind^{\Wrel\St}\,w^m\,w^f = \abs{\varphi~(s_1,s_2)}{w^m~(\abs{~((a_1,s'_1),(a_2,s'_2))}{w^f~(a_1,a_2)~(s'_1,s'_2)~\varphi})~(s_1,s_2)}.
\]
This construct is written in continuation passing style:
the specification of the continuation $w^f$  maps  a postcondition
$\varphi : (B_1 {\times} S_1) {\times} (B_2 {\times} S_2) {\to} \prop$
to an intermediate postcondition
$(A_1 {\times} S_1) {\times} (A_2 {\times} S_2) {\to} \prop$,
then $w^m$ turns it into a precondition for the whole computation.

%
Asynchronous rules for bind can be derived from the rule above, by taking $m_1$ to
be $\ret^{\Mone} ()$ or $f_1$ to be $\ret^{\Mone}$ above and using the monadic laws
of $\Mone$ (and symmetrically for $\Mtwo$):
\begin{align}
  \label{eq:bind-right-simple-rule}
  \inferrule{
    \srelsyn{\ret^\Mone\,()}{m_2}{w^m}\\
    \forall a_2 \srelsyn{c_1}{f_2\,a_2}{w^f\,a_2}
  }{
    \srelsyn{c_1}{\bind^\Mtwo\,m_2\,f_2}{\bind^\Wrel{}\,w^m\,(\abs{((), a_2)}{w^f\,a_2})}
  }
\end{align}
\begin{align}
  \inferrule{
    \srelsyn{c_1}{m_2}{w^m}\\
    \forall a_1, a_2 \srelsyn{\ret^\Mone\,a_1}{f_2\,a_2}{w^f\,(a_1,a_2)}
  }{
    \srelsyn{c_1}{\bind^\Mtwo\,m_2\,f_2}{\bind^\Wrel{}\,w^m\,w^f}
  }
\end{align}
Finally, for the effect-specific operations \ref{rulesetThree}, we provide a recipe for
writing rules guided by our framework.
For state, we introduce
the following asynchronous rules for
any $a_1,a_2$ and $s$:
\begin{align}
  \inferrule{ }{\srelsyn{\stget\,()}{\ret\,a_2}{w_{\stget^l}}} \qquad\qquad
  \inferrule{ }{\srelsyn{\ret\,a_1}{\stget\,()}{w_{\stget^r}}} \label{eq:get-simple-rules} \\
  \inferrule{ }{\srelsyn{\stput\,s}{\ret\,a_2}{w_{\stput^l}}} \qquad\qquad
  \inferrule{ }{\srelsyn{\ret\,a_1}{\stput\,s}{w_{\stput^r}}} \label{eq:put-simple-rules}
\end{align}
where $w_{\stget^l} = \abs{\varphi~(s_1,s_2)}{\varphi~((s_1,s_1),
(a_2, s_2))}$, $w_{\stget^r} = \abs{\varphi~(s_1,s_2)}{\varphi~((a_1,
s_1),(s_2,s_2))}$, $w_{\stput^l}
= \abs{\varphi~(s_1,s_2)}{\varphi~(((),s), (a_2, s_2))}$ and
$w_{\stput^r} = \abs{\varphi~(s_1,s_2)}{\varphi~((a_1,
s_1),((),s))}$.  Each of these rules describes at the specification
level the action of a basic stateful operation ($\stget{}$,
$\stput{}$) from either the left or the right computations, namely
returning the current state for $\stget{}$ or updating it for
$\stput{}$.
From these rules, we can derive two synchronous rules:
\begin{mathpar}
  \inferrule{ }{\srelsyn{\stget\,()}{\stget\,()}{w_{\stget}}}
  \and
  \inferrule{ }{\srelsyn{\stput\,s}{\stput\,s'}{w_{\stput}}}
\end{mathpar}
where $w_{\stget} =\abs{\varphi~s_1~s_2}{\varphi~((s_1,s_1), (s_2,
s_2))}$ and $w_{\stput} =\abs{\varphi~s_1~s_2}{\varphi~(((),s), ((),
s'))}$. These rules can be derived from the
\hyperref[eq:bindRuleSimple]{rule} for $\bind^{\Wrel{}}$, since by the monadic
equations we can replace for instance ${\srelsyn{\stget\,()}{\stget\,()}{ w_{\stget} }}$
by the following derivation
\begin{align*}
  \infer{\srelsyn{\ret\,()}{\stget\,()}{w_{\stget^l}} \\ \forall u:\one, s_2:S_2\enspace \srelsyn{\stget{}\,u}{\ret\,s_2}{w_{\stget^r}}}
  {\srelsyn{\bind^{\St_{S_1}}\enspace(\ret\,())\enspace\stget}{\bind^{\St_{S_2}}\enspace(\stget\,())\enspace\ret^{\St_{S_2}}}{\bind^\Wrel{\St}\enspace w_{\stget^l}\enspace(\abs{(u,s_2)}{w_{\stget^r}})}} 
\end{align*}
where the last specification reduces to $w_{\stget}$ using the definition of
$\bind^\Wrel{\St}$.
%

\paragraph{Simple semantics}
To define a semantics for the $\vdash$ judgment above, we generalize recent work
on (non-relational) effect observations \cite{dm4all} to the relational setting,
which raises significant challenges though.
We start from two ideas from the non-relational setting:
(1)~specifications
are drawn from a monad, ordered by precision \cite{nmb08htt, NanevskiBG13,
  ynot-icfp08, DelbiancoN13, fstar-pldi13, mumon, dm4free, dm4all} and (2)~one
can link any computation with its specification by defining a monad morphism,
\IE a mapping between two monads that respects their monadic structure.
In the case of state, an example monad morphism is
$\theta^\St(c) = \lambda \varphi~s.~ \varphi~(c~s) : \St\,A \to \W^\St\,A$,
mapping a stateful computation to the unary specification monad $\W^\St\,A = (A
{\times} S {\to} \prop) {\to} S {\to} \prop$, by running the computation and then checking whether
the postcondition holds of the result.
Inspired by \citet{Katsumata14}, \citet{dm4all}
call such monad morphisms {\em effect observations} and use them
to decouple the computational monad from the specification monad, which brings
great flexibility in choosing the specification monad and verification style
most suitable for the verification task at hand.
Intuitively, an effect observation accounts for the various choices available
when specifying computations with a particular effect, for instance total or
partial correctness, angelic or demonic nondeterminism, ghost state, \ETC
In this paper we bring this flexibility to the relational verification world.

For this, we observe that even though $\Wrel{}(A_1,A_2)$ is not a
monad, it is a {\em relative monad}~\cite{AltenkirchCU14}
over the
product $(A_1,A_2) \mapsto A_1 \times A_2$, as illustrated by the type
of $\bind^\Wrel{}$ above (\ref{eq:bindWrel-type}), where the
continuation specification is passed a pair of results from the first
specification.
%
%
Similarly, we generalize monad morphisms to relative monads and observe that
a relative monad morphism
$\theta_{\relmarker{}} : \Mone~A_1 \times \Mtwo~A_2 \to \Wrel{}(A_1,A_2)$
can immediately give us a semantics for the judgment above:
\[
\srelsem{c_1}{c_2}{w}{\theta_{\relmarker{}}} = \theta_{\relmarker{}}(c_1,c_2) \leq w,
\]
by asking that the specification obtained by $\theta_{\relmarker{}}$ is more precise than
the user-provided specification $w$.
In the case of state,
$\theta_{\relmarker{}}^\St(c_1,c_2) = \lambda \varphi~(s_1, s_2).~ \varphi~(c_1~s_1,c_2~s_2)$
simply runs the two computations and passes the results to the postcondition.
%
If we unfold this, and the definition of precision for $\Wrel{\St}$
\begin{align}
  \label{eq:WrelSt-order}
 w' \leq^\Wrel{\St} w = \forall \varphi\,s_1\, s_2.~ w\,\varphi\,(s_1,s_2)
  \Rightarrow w'\,\varphi\,(s_1,s_2),
\end{align}
we obtain the standard semantics of a relational program logic for stateful
computations (but without other side-effects):
\[
\srelsem{c_1}{c_2}{w}{\theta^\St_{\relmarker{}}} = \forall \varphi\,s_1\,s_2.~ w\,\varphi\,(s_1,s_2)
  \Rightarrow \varphi\,(c_1\,s_1,c_2\,s_2)
\]

Another important point is that the relational
effect observation can help us in deriving
simple effect-specific rules, such as the ones for $\stget{}$
(\ref{eq:get-simple-rules}) and $\stput{}$ (\ref{eq:put-simple-rules}) above.
For deriving such rules, one first has to choose $c_1$ and $c_2$ (and 
we hope that the product programs of \autoref{sec:product-programs}
can provide guidance on this in the future) and
then one can simply compute the specification using $\theta$.
For instance,
$w_{\stget^l} = \abs{\varphi~(s_1,s_2)}{\varphi~((s_1,s_1), (a_2, s_2))}$ in
the first get rule (\ref{eq:get-simple-rules}) really is just
$\theta(\stget\,(),\ret\,a_2)$.
This idea is illustrated for various other effects in \autoref{sec:effect-specific-rules}.

Finally, for probabilities and sometimes nondeterminism, one has to relax the
definition of relational effect observations to account for the fact that, in
the relational setting, modular verification can have a precision cost compared
to whole-program verification.
While for relative monad morphisms the following bind law has to hold with
equality (and analogously for returns):
\[
\theta_{\relmarker{}}\,(\bind^\Mone\,m_1\,f_1, \bind^\Mtwo\,m_2\,f_2) =
\bind^\Wrel{}\,(\theta_{\relmarker{}}\,(m_1,m_2))\,(\theta_{\relmarker{}}\circ (f_1,f_2))
\]
we introduce a notion of {\em lax} relative monad morphism that allows
the left-hand-side (\IE less modular verification)
to be more precise than the right-hand-side (\IE more modular verification):
\[
\theta_{\relmarker{}}\,(\bind^\Mone\,m_1\,f_1, \bind^\Mtwo\,m_2\,f_2) \leq
\bind^\Wrel{}\,(\theta_{\relmarker{}}\,(m_1,m_2))\,(\theta_{\relmarker{}}\circ (f_1,f_2)).
\]
For instance, the {\em refinement} relational effect observation
$\theta^{\forall\exists}$ interprets two nondeterministic computations $c_1 :
\NDet\,A_1$ and $c_2 : \NDet\,A_2$ (represented as finite sets of possible
outcomes) into the relational specification monad $\Wrel{\Pure}(A_1,A_2) = (A_1
\times A_2 \to \prop) \to \prop$
as follows:
\[
\theta^{\forall\exists}_{\relmarker{}}\,(c_1,c_2) =
  \lambda \varphi.~ \forall a_1 {\in} c_1.~ \exists a_2 {\in} c_2.~ \varphi (a_1,a_2)
\]
This interpretation is a natural generalization of subset inclusion (\IE refinement of
nondeterminism) to arbitrary relational postconditions $\varphi$, but only satisfies
the lax monad morphism law relating $\bind^\NDet~m~f = \bigcup_{a \in m}f\,a$
and $\bind^{\Wrel{\Pure}}~w~f= \lambda \varphi.~w~(\lambda(a_1,a_2).~ f~(a_1,a_2)~\varphi)$:
\[
\begin{array}{rcl}
\theta^{\forall\exists}_{\relmarker{}}\,(\bind^\NDet\,m_1\,f_1, \bind^\NDet\,m_2\,f_2) & = &
    \lambda \varphi.~ \forall a_1 {\in} m_1.~ \forall b_1 {\in} f_1\,a_1.~
                       \exists a_2 {\in} m_2.~ \exists b_2 {\in} f_2\,a_2.~ \varphi(b_1,b_2)\\
& \leq & \\
\bind^{\Wrel{\Pure}}(\theta^{\forall\exists}_{\relmarker{}}\,(m_1,m_2))\,(\theta^{\forall\exists}_{\relmarker{}}\circ (f_1,f_2)) & = &
  \lambda \varphi.~ \forall a_1 {\in} m_1.~ \exists a_2 {\in} m_2.~ \forall b_1 {\in} f_1~a_1.~ \exists b_2 {\in} f_2~a_2.~ \varphi(b_1,b_2)
\end{array}
\]
On the left-hand-side one can choose a different $a_2$ for every $b_1$, while on
the right-hand-side a single $a_2$ has to work for every $b_1$, so the two
preconditions are not logically equivalent.
Supporting such lax relational effect observations when needed is still
relatively simple, even if deriving useful effect specific rules is generally
more challenging in this case.

\paragraph{Exceptions, and why the simple semantics is not enough}
The simple construction we described so far is a natural
extension of the solution we previously proposed in the unary
setting~\cite{dm4all}. It works well for defining relational program logics
for state and nondeterminism (and
also input-output and probabilities), but it hits a limit when we
try to incorporate exceptions.
Indeed, defining a relational program logic for exceptions was an open research
problem, and our proposed solution depends on solving several non-trivial
technical challenges.
Here we begin with an analysis of the main obstruction of applying the simple
construction above to exceptions, and how that guides us to a generic
construction that can be made to work.

For relating computations that can raise exceptions, we often need to use
expressive specifications that can tell whether an exception was raised or not
in each of the computations.
For instance, such relational specifications could be drawn from:
\[
\Wrel{\Exc}(A_1,A_2) = ((A_1+E_1) \times (A_2+E_2) \to \prop) \to \prop.
\]
A predicate transformer $w : \Wrel{\Exc}(A_1,A_2)$ maps an exception-aware
postcondition $\varphi : (A_1+E_1) \times (A_2+E_2) \to \prop$ to a precondition,
which is just a proposition in $\prop$.
However, more work is needed to obtain
a \emph{compositional} proof system.
Indeed, suppose we have derivations for $\srelsyn{m_1}{m_2}{w^m}$ and $\forall a_1,a_2,
\srelsyn{f_1\,a_1}{f_2\,a_2}{w^f\,(a_1,a_2)}$ with specifications $w^m$
and $w^f\,(a_1,a_2)$ drawn from $\Wrel{\Exc}$.
In order to build a composite
proof relating $c_1 = \bind^\Exc\,m_1\,f_1$ and $c_2 = \bind^\Exc\,m_2\,f_2$ we need
compose $w^m$ and $w^f$ in some way.
If $w^m$ ensures that $m_1$ and $m_2$ terminate both normally returning values
we can compose with $w^f$ and if they both throw exceptions we can pass the
exceptions to the final postcondition.
Otherwise, a computation, say $m_1$, returns a value and the other, $m_2$,
raises an exception.
In this situation, the specification relating $c_1$ and $c_2$ needs a
specification for the continuation $f_1$ of $m_1$, but this cannot be
extracted out of $w^f$ alone.
%
In terms of the constructs of $\Wrel{\Exc}$, this failure is
an obstruction to complete the following tentative definition of
$\bind^{\Wrel\Exc}$:
\begin{lstlisting}
let bindWExc \,$w^m$ ($w^f$ : A_1 \times\, A_2 -> (((B_1 + E_1) \times\, (B_2 + E_2)) -> Prop) -> Prop) ($\varphi$ : (B_1 + E_1) \times\, (B_2 + E_2) -> Prop) =
  $w^m$ (fun$x$ : (A_1 + E_1) \times\, (A_2 + E_2).
      match $x$ with
      | Inl $a_1$, Inl $a_2$ -> $w^f$ $a_1$ $a_2$ $\varphi$
      | Inr $e_1$, Inr $e_2$ -> $\varphi$ (Inr $e_1$, Inr $e_2$)
      | _ -> ???  )
\end{lstlisting}
%
Our solution is to pass in two independent {\em unary} (\IE non-relational)
specifications for the continuations $f_1$ and $f_2$ as additional arguments for bind:
%
\begin{lstlisting}
let bindWExc \,$w^m$ ($\ii{w}^{f_1}$ : A_1 -> ((B_1 + E_1) -> Prop) -> Prop) ($\ii{w}^{f_2}$ : A_2 -> ((B_2 + E_2) -> Prop) -> Prop) $w^f$ $\varphi$ =
  $w^m$ (fun$x$ : (A_1 + E_1) \times\, (A_2 + E_2).
      match $x$ with
      ...
      | Inl $a_1$, Inr $e_2$ -> $w^{f_1}$ $a_1$ (fun$\ii{be}$. $\varphi$ $\ii{be}$ (Inr $e_2$))
      | Inr $e_1$, Inl $a_2$ -> $w^{f_2}$ $a_2$ (fun$\ii{be}$. $\varphi$ (Inr $e_1$) $\ii{be}$)  )
\end{lstlisting}
The first new \iffull match branch \else case \fi
corresponds to \iffull a situation where \else when \fi
$m_2$ terminated with an exception whereas $m_1$ returned a value normally.
In this situation, we use the unary specification $w^{f_2}$ to further evaluate
the first computation, independently of the second one, which already terminated.
The key observation is that the operation $\bind^\Wrel{\Exc}$ can still be used to define
a relative monad, but in a more complex relational setting that we introduce in
\autoref{sec:generic}.
As a consequence of moving to this more complex setting
our relational judgment needs to also keep track of unary
specifications, and its semantics also becomes more complex.
We tame this complexity by working this out internally to a {\em relational}
dependent type theory \cite{Ton13}.
These two novel conceptual ideas (combining unary and binary
specifications, and embedding inside a relational dependent type
theory for dealing with contexts) are fundamental pieces that
make the generic framework work.
In practice we can still implement this relational dependent type
theory inside our ambient type theory and continue using the same
tools for verification. We use \coq\ for
developing a proof of concept implementation: we represent the types in the
relational type theory using relations, and implement a set of
combinators that account for type formers such as $\Pi$-types.

\noindent
This paper makes the following \textbf{contributions}:
\begin{itemize}[leftmargin=*,nosep,label=$\blacktriangleright$]
\item We introduce the first generic framework for deriving relational program
  logics for arbitrary monadic effects, distilling the essence of previous
  relational program logics for specific effects.
  The proposed framework is highly expressive, and not only allows one to prove
  arbitrary relations between two programs with different monadic effects,
  but it also inherits the features of dependent type theory
  (higher-order logic, dependent types, polymorphism, lambdas, etc).

\item We provide a generic semantics for these relational program logics based
  on the novel observations that (1) the algebraic structure of relational
  specifications can be captured by particular relative monads, and (2) the two
  considered computations can be mapped to their specifications by
  particular relative monad morphisms we call relational effect observations.
  Our framework provides great flexibility in choosing the kind of relational
  specifications and the effect observation best suited for the verification
  task at hand.
  Finally, our generic rules are proven sound for any specification monad and
  any effect observation.


\item
  We show that this generic framework can be used to define relational
  program logics for effects as diverse as state,
  input-output, nondeterminism, and discrete probabilities.
%
%
  Moreover, we show that by instantiating our framework with state and
  unbounded iteration,
  we obtain a logic expressive enough to encode a variant of
  \citepos{benton04relational} Relational Hoare Logic (\RHL) (\autoref{sec:RHL}).
  Finally, we also sketch how \citepos{NanevskiBG13} Relational Hoare Type
  Theory (RHTT) can be reconstructed on top of our framework (\autoref{sec:RHTT}).

\item We identify and overcome conceptual challenges that prevented previous
  relational program logics from dealing with control effects such as exceptions
  \cite{BartheCK16}.
%
  We provide a proper semantic account of exceptions and the first relational
  program logic for this effect.



\item We propose a monadic notion
  of product programs, and illustrate it for the state effect.
\end{itemize}





\paragraph{Outline}
After recalling how computational monads can express a wide range of effects,
\autoref{sec:simplified} introduces relational specification monads and effect
observations, on top of which we build a simplified variant of our relational
reasoning framework, which we illustrate for state, input-output,
nondeterminism, discrete probabilities, and unbounded iteration,
and also with proofs of noninterference.
\autoref{sec:generic} then extends this simplified setting to account for
effects including exceptions, based on a relational dependent type theory
and also using relative monads as a unifying tool for the two settings.
\autoref{sec:embedRelProgramLogics} explains how to embed \RHL{} and the
connection to RHTT.
In \autoref{sec:product-programs} we discuss product programs,
before reviewing related work in \autoref{sec:related} and
concluding in \autoref{sec:conclusion}.
The ideas of this paper are supported by an accompanying \coq{}
development \ifanon in the anonymous supplementary material \fi
providing a proof of concept
implementation\ifanon\else~(available at \url{https://gitlab.inria.fr/kmaillar/dijkstra-monads-for-all/tree/relational})\fi~
that includes both the simplified and generic frameworks.



\section{Simplified Framework}
\label{sec:simplified}

In this section we introduce a simple framework for relational reasoning about
monadic programs based on (1) relational specification monads,
capturing relations between monadic programs, and (2) relational effect
observations, lifting a pair of computations to their specification.
By instantiating this framework with specific effects, we show how the specific
rules of previous relational program logics can be recovered in a principled way
and illustrate by example how these rules can be used to prove relational
properties of monadic programs, such as noninterference.
But first, we recall the monadic presentation of a few effects such as state,
exceptions, and nondeterminism.

\subsection{From Effects to Monads}
\label{sec:monads}

The seminal work of \citet{Moggi89} proposes using computational monads
to encapsulate effects. A monad is a parametrized type $\M\,A$
equipped with two operations $\ret^\M : A \to \M\,A$, sending a value
to an effectful context, and $\bind^\M : \M\,A \to (A \to \M\,B) \to
\M\,B$, sequentially composes an effectful computation returning values in $A$
with a continuation returning values in $B$, resulting in a $B$-valued
computation. Crucially, these operations obey 3 laws -- unitality of
$\ret$ with respect to $\bind$ and associativity of $\bind$ --
ensuring that any combinations of $\ret^\M$ and $\bind^\M$ can be seen
as a linear sequence of computations.
\begin{mathpar}
  \bindM\,(\retM\,a)~f = f\,a
  \and
  \bindM\,m~\retM = m
  \and
  \bindM\,m~(\abs{x}{\bindM\,(f\,x)~g}) = \bindM~(\bindM\,m\,f)~g
\end{mathpar}

A considerable number of effects are captured by monads, including stateful
computations, exceptions, interactive input-output,
nontermination,
%
nondeterminism, and continuations \cite{BentonHM00}.
Each monad comes with specific {\em operations} \cite{PP:NCDM} that allow the computation to
perform the actual effects that the monad provides. To
fix notation, we recall the basic monads corresponding to the effects
that we will use in the rest of the paper.

\paragraph{Stateful computations}
State passing functions $\St\,A = S \to A \times S$ are used to model
state, where $S$ is the type of the state.
The functions $\ret^\St$ and $\bind^\St$ are defined as
\begin{center}
  \lstinline!let retSt a : St A = fun s. (a,s)!\qquad
  \lstinline!let bindSt (m:St A) (f:A -> St B) : St B = fun s. let (a,s') = m s in f a s'!
\end{center}
This monad comes with two operations
\begin{center}
  \lstinline!let get : St S = fun s. (s,s)!\qquad\qquad
  \lstinline!let put (s:S) : St one = fun s_0. ((), s)!
\end{center}
that permit reading and updating the state.
A particular case of state are stores with many locations of a
particular type $\Val$. If $\mathcal{L}$ is a set of locations, then a
computations with a store of type $S = \mathcal{L} \to \Val$
are expressed by monad $\St_S$. In this case, we have custom
operations that are parameterized by the location which we are
accessing in the store:
\begin{center}
  \lstinline!let getL (l:loc) : St Val = fun s. (s l,s)!\qquad\qquad
  \lstinline!let putL (l : loc) (v : Val) : St one = fun s. ((), upd s l v)!
\end{center}
where \lstinline!let upd s l_1 v = fun l_2. if l_2 = l_1 then v else s l_2!.

\paragraph{Exceptions}
Computations potentially throwing exceptions of type $E$ are captured
by the type constructor $\Exc\,A = A + E$. The monadic operations are
\begin{lstlisting}
  let retExc a : Exc A = Inl a
  let bindExc (m:Exc A) (f:A -> Exc B) : Exc B = match m with | Inl a -> f a | Inr e -> Inr e
\end{lstlisting}
The operations provided are throwing and catching exceptions\footnote{Catching
  exceptions is the primary example of a \emph{handler} \cite{PlotkinP09}; we
  use here the term operation in a wide sense englobing both \emph{algebraic
    operations} (that we present as \emph{generic effects} \cite{PP:NCDM}) and handlers.}:
\begin{lstlisting}
let throw (e:E) : Exc zero = Inr e
let catch (m:Exc A) (mexc : E -> Exc A) : Exc A = match m with | Inl a -> Inl a | Inr e -> mexc e
\end{lstlisting}

\paragraph{Interactive Input-output}
Computations doing interactive input of type $I$ and output of type
$O$ are captured using monads as well. The type constructor has a
tree-like form
\begin{lstlisting}
type IO A = | Ret : A -> IO A | Input : (I -> IO A) -> IO A | Output : O -> IO A -> IO A
\end{lstlisting}
which consists of three possible cases: either we are done with a
return value (\lstinline!Ret!), or we expect a new input and then
continue (\lstinline!Input!), or we output and the continue
(\lstinline!Output!). The monadic function \lstinline!retIO!
constructs a unique leaf tree using \lstinline!Ret!
and \lstinline!bindIO! does tree grafting.
The operations perform input and output, and they are directly
captured using the corresponding constructors.
\begin{center}
\lstinline!let input : IO I = Input (fun i . retIO i)!\qquad\qquad
\lstinline!let output (o : O) : IO one = Output o (retIO ())!
\end{center}
We call this monad the input-output monad on $(I, O)$.


\paragraph{Nondeterminism}
The finite powerset $\NDet\,X = \mathcal{P}_\ii{fin}(X)$ models nondeterministic
computations as a set of possible outcomes. The return operation maps a value
$v$ to the singleton $\left\{ v \right\}$, and the bind operation uses union to
collect all results, \IE $\bind^\NDet~m~f = \bigcup_{v \in m}f\,v$. The
operation $\choice = \{ \true, \false \} : \NDet\,\bool$ nondeterministically
select a boolean value, whereas the operation $\fail{} : \NDet\,\zero$ does not
return any value. Repeating this operation, we can nondeterministically ${\choose
: (n:\bN) \to \fin\,n}$ an element of a finite set $\fin\,n$.

\paragraph{Imp-like effect}
To capture the syntax of simple imperative programs,
manipulating state and unbounded iteration, we introduce the $\Imp$ monad:
\begin{lstlisting}[mathescape]
type Imp A = | Ret : A -> Imp A $\hspace{2cm}$ | DoWhile : Imp bool -> Imp A -> Imp A
$\hspace{1.84cm}$| Get : (S -> Imp A) -> Imp A $\hspace{0.68cm}$| Put : S -> Imp A -> Imp A                    
\end{lstlisting}
Besides the monadic operations and the stateful ones, the $\Imp$ monad is built
to offer an operation
\begin{lstlisting}
let do_while (body : Imp bool) : Imp one  = DoWhile body (Ret ()) 
\end{lstlisting}
The expected semantics of this operation is to take a
computation \ls$body$ and to iterate \ls$body$ as long as it returns
true, so that the following equation -- which does not hold in
$\Imp$\ch{confusing, so where else it holds?}
-- is satisfied
\begin{center}
\lstinline!do_while body!\quad$=$\quad\lstinline!bindImp body (fun b . if b then do_while body else retImp ())!
\end{center}
When defining functions out of $\Imp$,
we will thus make sure that it holds in the target.



\paragraph{Probabilities}
A probabilistic computation is a sub-probability distribution on possible
outcomes, \IE{} for a countable type $A$, $\Prob\,A$ represents functions
$f : A \to \interval$ (where we write $\interval$ for the unit interval $[0;1]$)\ch{Why
  not simply use $[0;1]$ everywhere? For me
  $\interval$ doesn't really look like an I, but like a tiny little box.}
such that $\sum_{a\in A} f\,a \leq 1$.
Restricting our attention to discrete probabilities, the monad
structure on $\Prob$ is known as the \emph{Giry monad}~\cite{Giry}. 
The Dirac distribution at $v$ assigning weight $1$ to $v$ and $0$ to any other
value implements returns. Binding a distribution $m : \Prob\,A$ to a function
$f : A \to \Prob\,B$ amounts to computing the distribution on $B$ given by
$\abs{y}{\Sigma_{x \in \supp{m}} f\,x\,y}$.
We can consider various basic distributions on countable spaces as operations,
for instance $\flip : \interval \to \Prob\,\bool$ provides a Bernoulli
distribution $\mathcal{B}_p$ on booleans (with parameter $p \in \interval$)
.

\subsection{Specifications as (Relative) Monads}
\label{sec:simple-specs}

An important idea in the non-relational verification setting is to encapsulate the
specification of a monadic computation inside a monad \cite{nmb08htt, NanevskiBG13,
  ynot-icfp08, DelbiancoN13, fstar-pldi13, mumon, dm4free, dm4all},
giving the same algebraic footing to both computations and
specifications. For instance, stateful computations returning values
in $A$ are elements of a state monad $\St\,A = S \to (A \times S)$
and can be given specifications drawn from the monad $\W^\St\,A = (A
\times S \to \prop) \to (S \to \prop)$ equipped with the monad structure given
by
\begin{lstlisting}
let retWSt (a:A) : WSt A = fun $\varphi$ s. $\varphi$ (a,s)
let bindWSt (wm : WSt A) (wf : A -> WSt B) : WSt B = fun $\varphi$ s. wm (fun a. wf a $\varphi$) s
\end{lstlisting}
Intuitively, a specification $w : \W^\St\,A$ is a predicate transformer mapping
postconditions, which are predicates on the return value and final state, to
preconditions, which are predicates on the initial state. The monadic structure
on $\W^\St$ provides a canonical way to describe the monadic rules of a
non-relational program logic, \IE
\begin{align}
  \label{eq:dijstra-monads-typing-rules}
 \inferrule{\vdash v : A }{\vdash \ret^\St\,v: \St\,A~\{~\ret^{\W^\St}\, v~\}} 
  \qquad&\qquad
 \inferrule{\vdash m : \St\,A~\{~w^m~\} \\
   a:A\vdash f\,a : \St\,B~\{~w^f~\}}
 {\vdash \bind^\St\,m\,f : \St\,B~\{~\bind^{\W^\St}\,w^m\,w^f~\}}
\end{align}
This is, in fact, the main idea behind Dijkstra monads \cite{fstar-pldi13,
  mumon, dm4free, dm4all, Jacobs15}, which additionally internalize
$\St\,A~\{w\}$ as a computation type.

Now returning to the relational setting, a relational specification for a pair of
stateful computations $c_1 : \St_{S_1}\,A_1$ and $c_2 : \St_{S_2}\,A_2$
consists of a predicate transformer $w$ mapping postconditions
relating two pairs of a result value and a final state to a precondition relating
two initial states, \IE
\begin{align}
  \label{eq:srelspecState}
\Wrel{\St}(A_1,A_2)  =
  ((A_1 \times S_1) \times (A_2 \times S_2) \to \prop) \to
  S_1 \times S_2 \to \prop.
\end{align}
$\Wrel{\St}$ does not posses the monad structure its unary variant has.
To begin with it is not even an endofunctor: it
takes {\em two} types as input and produces one.
However, in order to derive a relational program logic, we need operations
playing the role of $\ret^{\W^\St}$ and $\bind^{\W^\St}$ in the unary
rules~(\ref{eq:dijstra-monads-typing-rules}).
In detail, we need a specification covering the case of two returns, as well as
a combinator producing a specification for a pair of $\bind^\St$ out of
specifications for the subcomputations.
In the particular case of $\Wrel{\St}$, the monadic operations of the unary
variant $\W^{\St}$ can be naturally extended to the relational setting providing such combinators:
\begin{lstlisting}[mathescape]
let retWrelSt (a_1,a_2):A_1 \times\, A_2 : WrelSt (A_1,A_2) = fun $\varphi$ (s_1, s_2). $\varphi$ ((a_1,s_1), (a_2,s_2))
\end{lstlisting}
\begin{lstlisting}[mathescape]
let bindWrelSt (wm : WrelSt (A_1,A_2)) (wf:A_1 \times\, A_2 -> WrelSt (B_1,B_2)) : WrelSt (B_1,B_2) =
  fun $\varphi$ (s_1,s_2). wm (fun ((a_1,s_1'),(a_2,s_2')). wf (a_1, a_2) $\varphi$ (s_1',s_2'))
\end{lstlisting}
These operations satisfy equations analogous to the monadic ones and are part of a
relative monad structure in the sense of \citet{AltenkirchCU14}. The relational
specifications for state $\Wrel{\St}$ are also naturally ordered by
$\leq^\Wrel{\St}$ (see (\ref{eq:WrelSt-order}) in \autoref{sec:intro})
and this ordering is compatible with the relative monad structure, as long as we
restrict our attention to \emph{monotonic} predicate transformers, a condition
that we will assume from now on for all monads on predicate transformers.
We call such a monad-like structure equipped with a compatible ordering a
\emph{simple relational specification monad}.
\begin{definition}
  \label{def:sRelSpecMon}
  A {\em simple relational specification monad} consist of
  \begin{itemize}
  \item for each pair of types $(A_1,A_2)$, a type $\Wrel{}(A_1,A_2)$
    equipped with a preorder $\leq^\Wrel{}$
  \item an operation $\ret^\Wrel{} : A_1 \times A_2 \to \Wrel{}(A_1,A_2)$
  \item an operation $\bind^\Wrel{} : \Wrel{}(A_1,A_2) \to (A_1\times A_2 \to
    \Wrel{}(B_1,B_2)) \to \Wrel{}(B_1,B_2)$ monotonic in both arguments
  \item satisfying the 3 following equations
    \begin{mathpar}
      \bind^\Wrel{}\,(\ret^\Wrel{}\,(a_1,a_2))~w^f = w^f\,(a_1,a_2)
      \and
      \bind^\Wrel{}\,w^m~\ret^\Wrel{} = w^m
      \and
      \bind^\Wrel{}\,(\bind^\Wrel{}\,w^m\,w^f)~w^g = \bind^\Wrel{} w^m (\abs{x}{\bind^\Wrel{}\,(w^f\,x)~w^g})
    \end{mathpar}
    for any $a_1 : A_1, a_2 : A_2, w^f : A_1 \times A_2 {\to}\Wrel{}(B_1,B_2),
    w^m:\Wrel{}(A_1,A_2), w^g :B_1\times B_2{\to}\Wrel{}(C_1,C_2)$.
  \end{itemize}
\end{definition}
A simple way to produce various examples of simple relational specification
monads besides $\Wrel{\St}$ is to start from a non-relational
specification monad $\W$ in the
sense of \citet{dm4all}, that is a monad equipped with a compatible order, and
to compose it with the function $(A_1,A_2)\mapsto A_1 \times A_2$. A result of
\citet{AltenkirchCU14} (prop. 2.3.(1)) then ensures that $\Wrel{}(A_1,A_2) =
\W(A_1\times A_2)$ is a simple relational specification monad. In the following
paragraphs, we illustrate this construction with a few concrete instances
showing the flexibility of this construction. Depending on the property we want to
verify and the desired verification style, we can pick relational specification
monads among many different alternatives.
For instance, choosing a simpler relational specification monad can often
simplify verification, but also have less expressive power than more
sophisticated variants.
Similarly, relational weakest preconditions are better suited for
(semi-)automated verification, but relational pre-/postconditions are more
intuitive to humans and make the connection to established relational program
logics more evident.

\paragraph{Backward predicate transformer} A stateless version
of $\Wrel{\St}$ is the \iffull backward \fi predicate transformer
\[\Wrel{\Pure}(A_1,A_2) = (A_1 \times A_2 \to \prop) \to \prop\]
equipped with monadic operations and order derived from the monotonic
continuation monad. We call this simple relational specification monad
$\Pure$ because it naturally applies to the relational verification of pure code,
however it can also be useful to verify effectful code as we will see
for nondeterministic computations in
\autoref{sec:effect-specific-rules}.


\paragraph{Pre-/postconditions}
Specifications written in terms of pre-/postconditions are simpler to
understand than their predicate transformer equivalents. We
show that relational specifications written as pre-/postcondition
also form a relational specification monad. The type constructor
\[
\PPrel{\Pure}(A_1, A_2) = \prop \times (A_1 \times A_2 \to \prop)
\]
models a pair consisting of a precondition in $\prop$ and a postcondition, that
is a relation on final values of two computations. There is a natural ordering between such
pairs, namely
\[(\pre_1, \post_1) \leq^{\PPrel{\St}} (\pre_2, \post_2) \quad \iff\quad
    \begin{array}{l}
      \pre_2\Rightarrow \pre_1 \enspace \wedge\\
      \forall (a_1 : A_1)(a_2 : A_2). \post_1 (a_1, a_2) {\Rightarrow} \post_2 (a_1, a_2).   
    \end{array}
\]
The monadic structure is given by
\begin{lstlisting}[mathescape]
let retPP (a_1, a_2) = ( True,\,\, fun (a_1', a_2'). a_1 = a_1' /\ a_2 = a_2' )
let bindPP (pre, post) f = 
  let pre' = pre /\ forall a_1, a_2 . post (a_1, a_2) ==> $\pi$_1 (f (a_1, a_2)) in
  let post' (b_1, b_2) = exists a_1, a_2 . post (a_1, a_2) /\ $\pi$_2 (f (a_1, a_2)) (b_1, b_2) in
  (pre', post')
\end{lstlisting}
The return operation results in a trivial precondition and a postcondition
holding exactly for the given arguments, whereas $\bind^{\PPrel{\Pure}}$
strengthens the precondition of its first argument so that the postcondition of the
first computation entails the precondition of the continuation.

\paragraph{Stateful pre-/postconditions}
Continuing on pre-/postconditions, we consider a stateful variant of
$\PPrel{\Pure}$:
\[
\PPrel{\St}(A_1,A_2) = (S_1 \times S_2 \to \prop) \times ((S_1 \times A_1 \times S_1) \times (S_2 \times A_2 \times S_2) \to \prop)
\]
These are pairs, where the first component consists of a precondition
on a pair of initial states, one for each sides, while the second component is a
postcondition formed by a relation on triples of an initial state, a final value
and a final state.

The simple relational monadic specification
structure is similar to the one of $\PPrel{\Pure}$, threading in the state where
necessary, and specifying that the initial state does not change for return:
\begin{center}
\lstinline[mathescape]!let retPPrelSt (a_1,a_2) = (fun (s_1, s_2) . True,\,\,\,\,\,\, fun (($s^i_1$, a_1', $s^f_1$),($s^i_1$, a_2', $s^f_2$)) . a_1 = a_1' /\ a_2 = a_2' /\ $s^i_1$ = $s^f_1$ /\ $s^i_2$ = $s^f_2$)!.
\end{center}

There is a natural embedding of stateful pre-/postconditions $(\pre, \post):
\PPrel{\St}(A_1,A_2)$ into stateful backward predicate transformers
$\Wrel{\St}(A_1,A_2)$ given by
\[
\abs{\varphi~(s^i_1,s^i_2)}{\pre (s^i_1, s^i_2) \wedge \forall a_1, a_2, s^f_1, s^f_2 . \post~((s^i_1, a_1, s^f_1),(s^i_2, a_2, s^f_2)) {\Rightarrow} \varphi~((a_1,s^f_1),(a_2,s^f_2))} : \Wrel{\St}(A_1,A_2).
\]


\paragraph{Errorful backward predicate transformer}
While exceptions turn out to be complex in general, a coarse approach is
still possible
 using the simple relational monad 
\begin{align}
  \label{eq:wrelerr}
\Wrel{\Err}(A_1,A_2) = ((A_1\times A_2 + \one) \to \prop) \to \prop.
\end{align}
This construction represents a predicate transformer that works on
either successful computations, or on an indication that at least one of
the computations threw an exception, but losing the information of which of the
two sides raised the exception. We can actually show that, under mild assumptions,
no simple relational specification monad accounting for exceptions can
distinguish the three situations where the left, the right, or both programs are
raising exceptions. Intuitively, this is due to the fact that the
two programs are supposed to run independently, but the simple relational
specification monad impose some amount of synchronization. We return to
$\Wrel{\Exc}$ from \autoref{sec:intro}
and solve this problem in \autoref{sec:generic}, while previous
relational program logics have generally been stuck
with weak specification monads in the style of
$\Wrel{\Err}$ above \cite{BartheCK16}.

\paragraph{Input-output backward predicate transformer}
A relational specification monad similar to $\Wrel{\St}$ can be used
to specify interactive I/O computations. For
relating two computational monads on input-output sets $(I_1, O_2)$
and $(I_2, O_2)$, we use
\begin{align}
  \label{eq:wrelio}
\Wrel{\IO}(A_1, A_2) = (A_1 \times A_2 \to \mathrm{list}(\mathcal{E}_1) \times \mathrm{list}(\mathcal{E}_2) \to \prop) \to \mathrm{list}(\mathcal{E}_1) \times \mathrm{list}(\mathcal{E}_2) \to \prop
\end{align}
where $\mathcal{E}_1 = I_1 + O_1$ and $\mathcal{E}_2 = I_2 + O_2$
represent a log element of possible input-output
behaviour. Intuitively, a specification of type $\Wrel{\IO}(A_1, A_2)$
is a backward predicate transformer that transforms a postcondition on
the output results and the I/O history into a precondition describing
the I/O history before running the computations. Alternative relational
specification monads for input-output are easily defined following the
discussion in~\citet{dm4all}.

\paragraph{Quantitative predicate transformers} The backward predicate
transformer $\Wrel{\Pure}$ generalizes to quantitative settings were
propositions are replaced by a notion of resource. We use a
particular case of this generalization as the relational specification
monad for probabilities, restricting to monotonic additive
continuous\footnote{as maps between $\omega$-cpo} maps in the type~\cite{AP-M:2006,Fai17}
\[\Wrel{\Prob}(A_1,A_2) = (A_1 \times A_2 \to \interval) \to \interval \]

\ch{Shouldn't we also cite Alejandro's work on weakest pre-expectations~\cite{AguirreBHKKM19}?}

\subsection{Relational Semantics from Effect Observations}
\label{sec:effect-observations}


The relational judgment $\srelsyn{c_1}{c_2}{w}$ should assert that monadic
computations $c_1:\Mone{}A_1$ and $c_2:\Mtwo{}A_2$ satisfy a relational
specification  $w:\Wrel{}(A_1,A_2)$ drawn from a simple relational specification
monad. What does this judgment mean in our semantic framework? Certainly it
requires a specific connection between the computational monads $\Mone{}$,
$\Mtwo{}$ and the simple relational specification monad $\Wrel{}$.
In the non-relational setting, this is accomplished by an \emph{effect
  observation}, \IE a monad morphism from the computational monad to the
specification monad \cite{Katsumata14, dm4all}.
An effect observation accounts for the various choices available when specifying
a particular effect, for instance total or partial correctness in the case of
errors or recursion, angelic or demonic interpretations of nondeterministic
computations, or connecting ghost state with actual state or with past IO events.
In the relational setting, we introduce \emph{relational effect observations},
families of functions respecting the monadic structure, defined here
from first principles, but arising as an extension of monad morphisms
as we will show in \autoref{sec:RelMon}.

\begin{definition}
  \label{def:sRelEffObs}
  A {\em simple lax relational effect observation} $\theta_{\relmarker{}}$ from computational monads $\Mone,\Mtwo$ to a
  simple relational specification monad $\Wrel{}$ is given by
  \begin{itemize}
  \item for each pair of types $A_1,A_2$ a function $\theta_{\relmarker{}} : \Mone\,A_1 \times
    \Mtwo\,A_2 \to \Wrel{}(A_1,A_2)$
  \item such that
    \begin{align*}
      \theta_{\relmarker{}}\,(\ret^\Mone\,a_1,\ret^\Mtwo\,a_2) &\leq^{\Wrel{}} \ret^{\Wrel{}}\,(a_1,a_2)\\
      \theta_{\relmarker{}}\,(\bind^\Mone\,m_1\,f_1, \bind^\Mtwo\,m_2\,f_2) &\leq^{\Wrel{}}
      \bind^\Wrel{}\,(\theta_{\relmarker{}}\,(m_1,m_2))\,(\theta_{\relmarker{}}\circ (f_1,f_2))
    \end{align*}
  \end{itemize}
  We say that $\theta_{\relmarker}$ is a {\em simple strict relational effect
    observation} if these two laws hold with equality.
\end{definition}

As explained in the introduction,
for stateful computations 
a simple strict relational effect observation targeting $\Wrel{\St}$
runs the two computations and passes the results to the
postcondition:
\begin{align}
  \label{eq:state-effect-observation}
  \theta^\St_{\relmarker{}}(c_1,c_2) = \abs{\varphi\,(s_1,s_2)}{\varphi(c_1\,s_1, c_2\,s_2)}.
\end{align}
A more interesting situation happens when interpreting nondeterministic
computations $(c_1,c_2):\NDet{}\,A_1\times \NDet{}\,A_2$ into the relational specification monad $\Wrel{\Pure}(A_1,A_2)$.
Two natural simple strict relational effect observations are given by
\begin{align}
  \label{eq:nd-effect-observation}
  \theta_{\relmarker{}}^\forall(c_1,c_2) &= \abs{\varphi}{\forall a_1\in c_1, a_2 \in c_2.\> \varphi(a_1,a_2)},&
  \theta_{\relmarker{}}^\exists(c_1,c_2) &= \abs{\varphi}{\exists a_1\in c_1, a_2 \in c_2.\> \varphi(a_1,a_2)}.
\end{align}
The first one $\theta_{\relmarker{}}^\forall$ prescribes
that all possible results from the left and right computations have to
satisfy the relational specification, corresponding to a demonic
interpretation of nondeterminism, whereas the angelic
 $\theta_{\relmarker{}}^\exists$ requires at least one final value on each
sides to satisfy the relation.

These examples are instances of
the following theorem, which allows to lift unary effect observations
to simple strict relational effect observations.
To state it, we first recall that two monadic computations $c_1 : \M\,A_1$ and
$c_2 : \M\,A_2$ \emph{commute}
%
\cite{Fuhrmann02,BowlerGLS13}
when
\[
\bind^\M~c_1~\left(\abs{a_1}{\bind^\M~c_2~\left(\abs{a_2}{\ret^\M (a_1, a_2)}\right)}\right) =
\bind^\M~c_2~\left(\abs{a_2}{\bind^\M~c_1~\left(\abs{a_1}{\ret^\M (a_1, a_2)}\right)}\right).
\]
The intuition is that executing $c_1$ and then $c_2$ is the same as
executing $c_2$ and then $c_1$.
\begin{theorem}
\label{thm:effect-commutation}
  Let $\theta_1 : \M_1 \to \W$ and $\theta_2 : \M_2 \to \W$ be unary
  effect observations, where $\M_1$ and $\M_2$ are computational
  monads and $\W$ is a (unary) specification monad. We denote with
  $\Wrel{}(A_1,A_2) = \W\,(A_1\times A_2)$ the simple strict relational specification
  monad derived from $\W$ (see~\autoref{sec:simple-specs}\ch{where exactly?}). If for all
  $c_1 : \M_1\,A_1$ and $c_2 : \M_2\,A_2$, we have that
  $\theta_1(c_1)$ and $\theta_2(c_2)$ commute, then the following
  function $\theta_{\relmarker{}} : \Mone{}\,A_1 \times \Mtwo{}\,A_2 \to \Wrel{}(A_1,A_2)$ is a simple
  relational effect observation
\[
  \theta_{\relmarker{}}(c_1, c_2)
  = \bind^\W~\theta_1(c_1)~\left(\abs{a_1}{\bind^\W~\theta_2(c_2)~\left(\abs{a_2}{\ret^\W
(a_1, a_2)}\right)}\right).
\]
\end{theorem}
Moreover, a partial converse for this theorem holds: given a simple
relational effect observation $\theta_{\relmarker}
: \Mone{}\,A_1 \times \Mtwo{}\,A_2 \to \Wrel{}(A_1,A_2)$ where
$\Wrel{}$ is a lifting of a unary specification monad (\IE
$\Wrel{}(A_1, A_2) = \W(A_1 \times A_2)$), then there exist commuting
unary effect observations $\theta_1 : \Mone{} \to \W$ and $\theta_2
: \Mtwo{} \to \W$ such that $\theta_{\relmarker}$ is equal to the
simple relational effect observation obtained from
applying \autoref{thm:effect-commutation} to these.

Another class of examples of {\em lax} effect observations,
covering for instance the refinement observation for nondeterminism $
\theta^{\forall\exists}_{\relmarker{}}\,(c_1,c_2) = \lambda \varphi.~ \forall
a_1 {\in} c_1.~ \exists a_2 {\in} c_2.~ \varphi (a_1,a_2)$ from \autoref{sec:intro},\ch{
  This subsection doesn't really show how to do this example}
is provided by the following theorem that connects \emph{lax} effect observations to
\emph{relators} $\Gamma$ over the monad $M$~\citep{LagoGL17, Gavazzo18} which lift
relations on values to relations on monadic computations:
\[
\Gamma \quad:\quad (A_1 \times A_2 \to \prop) \longrightarrow M A_1 \times M A_2 \to \prop.
\]
\begin{theorem}
  A relator $\Gamma$ over a monad $M$ induces a simple lax relational
  effect observation of the form $\theta^\Gamma_{\relmarker{}} : M
  A_1 \times M A_2 {\to} \Wrel{\Pure}(A_1, A_2)$.
\end{theorem}
\begin{proof}
  The carrier of $\theta^\Gamma_{\relmarker{}} : M\,A_1\times M\,A_2 \to (A_1
  \times A_2 \to \prop) \to \prop$ is obtained by swapping the arguments of
  $\Gamma$, while the two inequalities are direct consequences of the compatibility
  of the relator $\Gamma$ with the monad $M$.
\end{proof}
Relators provide interesting examples of relational effect observations
for nondeterminism and for probabilities. Relational effect observation extend
relators by providing the possibility of relating two different computational
monads, as well as having more sophisticated specifications with ghost state or
exceptional postconditions. Conversely, relators preserve -- in a lax sense --
identities and relational composition.\ch{Is this really a ``conversely'',
  given that relators are just a special case?}\ch{I anyway don't understand this claim
  and how it's supposed to show a connection to relational effect observations.}

\km{Should explain that this result provides us other examples too}

In general, given a simple lax relational effect observation $\theta_{\relmarker{}} :
\Mone{},\Mtwo{} \to \Wrel{}$,
we define the semantic relational judgment by 
\begin{align}
  \label{eq:simple-semantic-judgement-def}
 \srelsem{c_1}{c_2}{w}{\theta_{\relmarker{}}}\quad=\quad\theta_{\relmarker{}}\,(c_1,c_2) \leq^\Wrel{} w, 
\end{align}
where we make use of the preorder given by $\Wrel{}$.
The following 3 subsections explain how to derive sound rules for a relational logic
parameterized by the computational monads $\Mone{}, \Mtwo{}$, the simple
relational specification monad $\Wrel{}$, and the simple lax relational effect
observation $\theta_{\relmarker{}}$.

\subsection{Pure Relational Rules}
\label{sec:simple-pure-relational-rules}

We start with rules coming from the ambient dependent type theory. Even though
the semantics of the relational judgment depends on the choice of an
effect observation, the soundness of the basic pure rules introduced in
\autoref{fig:pure-relational-rules} is independent from both the computational
monads and effects observation.
Indeed, the proof of soundness of these follows from applying the adequate
dependent eliminator coming from the type theory.
\begin{figure}[h]
\begin{mathpar}
  \inferrule*[left=$\bool$-Elim]{\Ifte{b}{\srelsyn{c_1}{c_2}{w^\top}}{\srelsyn{c_1}{c_2}{w^\bot}}}
  {\srelsyn{c_1}{c_2}{\Ifte{b}{w^\top}{w^\bot}}}
  \and
  \inferrule*[left=$\zero$-Elim{\footnote[2]{}}]{w \leq \dot{\bot}}
  {\srelsyn{c_1}{c_2}{w}}
  \and
  \inferrule*[left=$\bN$-Elim]{
    n : \bN \\ w = \kwd{elim}^\bN\;w_0\;w_{suc}\\
    \srelsyn{c_1[0/n]}{c_2[0/n]}{w_0}\\
    \forall n:\bN,~\srelsyn{c_1}{c_2}{w\,n} \quad\Rightarrow\quad \srelsyn{c_1[\kwd{S}\,n/n]}{c_2[\kwd{S}\,n/n]}{w_{suc}\,(w\,n)}}
  {\srelsyn{c_1}{c_2}{w\;n}}
\end{mathpar}
  \caption{Pure relational rules}
  \label{fig:pure-relational-rules}
\end{figure}
\ch{TODO: The rules of \autoref{fig:pure-relational-rules} should be explained}

These rules can then be tailored as explained in the introduction to derive
asynchronous (\ref{eq:simple-rule-if-left}) or synchronous
(\ref{eq:simple-rule-if-sync}) rules more suited for applications.
For some of the derived rules,
there is, however, an additional
requirement on the simple relational specification monad,
so that we can strengthen preconditions. This small mismatch in the theory,
already present in the unary setting of \citet{dm4all} on top of which
we work, could be solved by adopting a richer definition of
specification monads, for instance taking inspiration in the work of
\citet{Gavazzo18}, and is left as future work.
\ch{So how far are we now from actually solving it?
  It's the kind of statement that weakens our contribution.}
\ch{Basically half(!) of this subsection is spent bashing our own work.
  It would be a lot more productive to explain it first, and also explaining why
  being in a dependent type theory is cool.}

\subsection{Generic Monadic Rules}
\label{sec:simple-generic-monadic-rules}
\footnotetext[2]{Assuming that $\Wrel{}$
      contains a top element $\dot{\bot}$ that entails falsity of the
      precondition; this is the case for all our examples.}
Given any computational monads $\Mone, \Mtwo$ and a simple relational
specification monad $\Wrel{}$, we introduce three rules governing the monadic
part of a relational program logic (\autoref{fig:smonRules}).
\begin{figure}[h]
\begin{mathpar}
  \inferrule*[left=Ret]{ a_1: A_1\\ a_2 : A_2 }{\srelsyn{\ret^\Mone\,a_1}{\ret^\Mtwo\,a_2}{\ret^\Wrel{}\,(a_1,a_2)}}
  \and
  \inferrule*[left=Weaken] {\srelsyn{c_1}{c_2}{w}\\ w\leq w'}
  {\srelsyn{c_1}{c_2}{w'}}
  \and
  \inferrule*[left=Bind]{\srelsyn{m_1}{m_2}{w^m}\\
    \forall a_1,a_2
    \srelsyn{f_1\,a_1}{f_2\,a_2}{w^f\,(a_1,a_2)}}
  {\srelsyn{\bind^\Mone\,m_1\,f_1}{\bind^\Mtwo\,m_2\,f_2}{\bind^\Wrel{}\,w^m\,w^f}}
\end{mathpar}
  \caption{Generic monadic rules in the simple framework}
  \label{fig:smonRules}
\end{figure}
Each of these rules directly corresponds to one aspect of the
simple relational specification monad and are all synchronous.
As explained in the introduction (\ref{eq:bind-right-simple-rule}), it is then
possible to derive asynchronous variants using the monadic laws of the
computational monads.
\begin{theorem}[Soundness of generic monadic rules]
\label{thm:simple-generic-soundness}
 The relational rules in \autoref{fig:smonRules} are sound with respect to any
 lax relational effect
 observation $\theta_{\relmarker{}}$, that is 
 \[\srelsyn{c_1}{c_2}{w} \qquad{\Rightarrow{}}\qquad \forall\theta_{\relmarker{}},~\srelsem{c_1}{c_2}{w}{\theta_{\relmarker{}}}.\]
\end{theorem}
\begin{proof}
  For rules \textsc{Ret} and \textsc{Bind}, we need to prove that
  $\theta_{\relmarker{}}(\ret^\Mone{}\,a_1,\ret^\Mtwo{}a_2) \leq \ret^\W{}(a_1,a_2)$ and
  $\theta_{\relmarker{}}\,(\bind^\Mone\,m_1\,f_1, \bind^\Mtwo\,m_2\,f_2) \leq
  \bind^\W\,(\theta_{\relmarker{}}\,(m_1,m_2))\,(\theta_{\relmarker{}}\circ
  (f_1,f_2))$, which are exactly the laws of a lax relational effect observation.
  For \textsc{Weaken}, we need to show that $\theta_{\relmarker{}}(c_1,c_2)
  \leq w'$ under the assumptions that $\theta_{\relmarker{}}(c_1,c_2)\leq w$ and $w \leq
  w'$ so we conclude by transitivity.
\end{proof}

\subsection{Effect-Specific Rules}
\label{sec:effect-specific-rules}
The generic monadic rules together with the rules coming from the
ambient type theory allow to derive relational judgments for the main
structure of the programs. However, these rules are not enough to
handle full programs written in the computational monads $\Mone{}$ and
$\Mtwo{}$, as we also need rules to reason about the specific effectful
operations that these monads provide.
The soundness of effect-specific relational rules is established with
respect to a {\em particular} choice of relational
effect observation $\theta_{\relmarker{}} : \Mone{},\Mtwo \to \Wrel{}$.
Consequently, we make essential use of $\theta_{\relmarker{}}$ to devise
effect-specific rules.
The recipe was already illustrated for state in the introduction: first pick a
pair of effectful {\em algebraic} operations (or $\ret$ for the asynchronous rules),
unfold their definition, and then compute a sound-by-design relational
specification for this pair by simply applying $\theta_{\relmarker{}}$.
%
%
By following this recipe, we are decoupling the problem of choosing the
computations on which these rules operate (\EG~synchronous vs.  asynchronous
rules to which we return in \autoref{sec:product-programs}) from the problem of
choosing sensible specifications, which is captured in the choice of
$\theta_{\relmarker{}}$.

%




\paragraph{Nondeterministic computations}
The two relational effect observations $\theta_{\relmarker{}}^\forall$ and
$\theta_{\relmarker{}}^\exists$ provide different relational rules for the operation
$\choice{}$. As an example of how the recipe works, suppose that we
want to come up with an asymmetric rule for nondeterministic
computations that works on the left program, and which is sound with
respect to $\theta_{\relmarker{}}^\forall$. This means that the conclusion will be of
the form $\srelsyn{\choice{}}{\ret{}\,a_2}{w_{\choice^l}}$ for some $w_{\choice^l}
: \PPrel{\Pure}$. To obtain $w_{\choice^l}$, we apply the effect observation to the
two computations involved in the rule
\[
w_{\choice^l} = \theta_{\relmarker{}}^\forall\left(\choice{},\ret{}\,a_2\right)
= \abs{\varphi}{\forall b \in \{ \true, \false \}, a \in \{ a_2 \}.~ \varphi(b, a)}
= \abs{\varphi}{\varphi(\true, a_2) \land \varphi(\false, a_2)},
\]
obtaining a rule that is trivially sound:
\begin{mathpar}
  \inferrule*[left=DemonicPickLeft]{ }{\srelsyn{\choice{}}{\ret{}\,a_2}{\abs{\varphi}{\varphi(\true,a_2) \wedge \varphi(\false, a_2)}}}.
\end{mathpar}
Similarly for \fail, we compute $w_{\fail^l}$ for $\srelsyn{\fail{}}{\ret{}\,a_2}{w_{\fail^l}}$ as follows:
\[
w_{\fail^l} = \theta_{\relmarker{}}^\forall\left(\fail{},\ret{}\,a_2\right) = \abs{\varphi}{\forall b \in \{ \true, \false \}, a \in \emptyset.~ \varphi(b, a)} = \abs{\varphi}{\top},
\]
\begin{mathpar}
  \inferrule*[left=DemonicFailLeft]{ }{\srelsyn{\fail{}}{\ret{}\,a_2}{\abs{\varphi}{\top}}}.
\end{mathpar}
Following the same approach, we can come up with an asymmetric rule
on the right as well as a symmetric one. For concreteness, we show
the symmetric rule for the effect observation $\theta_{\relmarker{}}^\exists$:
\begin{mathpar}
  \inferrule*[left=Angelic]{ }{\srelsyn{\choice{}}{\choice{}}{\abs{\varphi}{\varphi(\true,\true) \vee \varphi(\true, \false) \vee \varphi(\false, \true) \vee \varphi(\false, \false)}}}.
\end{mathpar}
\ch{TODO: It would also be good to give the rules for $\theta^{\forall\exists}_{\relmarker{}}$}
Taking inspiration from the sample rule in~\cite{BartheEGHSS15}, we introduce a
rule for the refinement effect observation
$\theta^{\forall\exists}_{\relmarker{}}$ using an auxilliary function to select the
elements in correspondence:
\begin{mathpar}
  \inferrule*[left=Refinement]{ h : \fin\,n \to \fin\,
    m}{\srelsyn{\choose{}\,n}{\choose{}\,m}{\abs{\varphi}{\forall k. \varphi(k, h\,k)}}}.
\end{mathpar}





\paragraph{Exceptions using $\Wrel{\Err}$}
Taking $\Mone{}$ and $\Mtwo{}$ to be exception monads on exception sets $E_1$
and $E_2$, and the relational specification monad $\Wrel{\Err}$
(\autoref{eq:wrelerr} on page \pageref{eq:wrelerr}), we have an
effect observation interpreting any thrown exception
as a unique erroneous termination situation, that is
\begin{lstlisting}
let $\theta_{\relmarker}^{\Err}$ ((c_1, c_2) : Exc A_1 \times Exc A_2) : WrelErr (A_1,A_2) =
 fun $\varphi$.  match c_1, c_2 with | Inl a_1, Inl a_2 -> $\varphi$ (Inl (a_1, a_2)) | _, _ -> $\varphi$ (Inr ()) 
\end{lstlisting}
Under this interpretation we can show the soundness of the following rules:


\newsavebox\catchspecbox
\savebox\catchspecbox{\lstinline[mathescape]!fun $\varphi$. $w$ (fun a_0. match a_0 with Inl a -> $\varphi\hspace{2pt}$(Inl a) | Inr () -> $w^{\maltese}\,\varphi$)!}
\begin{mathpar}
  \inferrule*[left=ThrowL]{ }{\srelsyn{\throw\,e_1}{\ret\,a_2}{\abs{\varphi}{\varphi(\inr\,())}}}
  \and
  \inferrule*[left=ThrowR]{ }{\srelsyn{\ret\,a_1}{\throw\,e_2}{\abs{\varphi}{\varphi(\inr\,())}}}
  \and
  \inferrule*[left=Catch]{
    \srelsyn{c_1}{c_2}{w}\\
    \forall e_1\,e_2\srelsyn{c^{\maltese}_1\,e_1}{c^{\maltese}_2\,e_2}{w^\maltese}\\
    \forall e_1\,a_2 \srelsyn{c^{\maltese}_1\,e_1}{\ret\,a_2}{w^\maltese}\\
    \forall a_1\,e_2 \srelsyn{\ret\,a_1}{c^{\maltese}_2\,e_2}{w^\maltese}
  }{\srelsyn{\catch\,c_1\,c^{\maltese}_1}{\catch\,c_2\,c^{\maltese}_2}
    {\usebox{\catchspecbox}}}
\end{mathpar}
The rules \textsc{ThrowL} and \textsc{ThrowR} 
can be derived using the recipe above,
but the exceptions have to be
conflated to the same exceptional result $\inr\,()$, a situation that is forced
by the choice of relational effect observation and a weak specification monad.
As a consequence, the \textsc{Catch} rule considers one successful
case and three exceptional cases. The specification in the conclusion
takes a postcondition $\varphi$ and computes a precondition by running
the transformer $w$ on a new postcondition that depends on the result
of $c_1$ and $c_2$. If both computations were successful, then
this new postcondition is simply the original $\varphi$. If an
exception was thrown (in any of the sides or both), then the new
postcondition is computed using the transformer $w^{\maltese}$,
which specifies the three exceptional cases.
The specification for \textsc{Catch} does not follow mechanically from
$\theta^\Err_{\relmarker{}}$ using our recipe, since it is a handler and not an
algebraic operation.
\ch{In our last discussion it seemed unclear if this is the best rule.
  In the complex setting we have a different rule?}

\paragraph{Input-output computations}
Let $\Mone{}$ and $\Mtwo{}$ be the input-output monads
on $(I_1, O_1)$ and $(I_2, O_2)$ respectively (\autoref{sec:monads}).
We want an effect observation on the relational
specification monad $\Wrel{\IO}$ (\autoref{eq:wrelio} on
page \pageref{eq:wrelio}):
\[
\theta^{\IO}_{\relmarker{}} : \Mone~A_1 \times \Mtwo~A_2 \to \Wrel{\IO}(A_1,A_2)
\]
Notice that $\Wrel{\IO}(A_1, A_2) = \W^{\IO}(A_1 \times A_2)$, where $\W^{\IO}$ is a unary specification monad defined by
\[
\W^{\IO}(A) = (A \to \mathrm{list}(\mathcal{E}_1) \times \mathrm{list}(\mathcal{E}_2) \to \prop) \to \mathrm{list}(\mathcal{E}_1) \times \mathrm{list}(\mathcal{E}_2) \to \prop
\]
By applying \autoref{thm:effect-commutation} to unary effect
observations $\theta^{\IO}_1 : \Mone{} \to \W^{\IO}$ and
$\theta^{\IO}_2 : \Mtwo{} \to \W^{\IO}$, we obtain the desired relational effect
observation $\theta^{\IO}_{\relmarker{}}$. The unary effect
observation $\theta^{\IO}_1$ is defined by recursion on the
computation trees ($\theta^{\IO}_2$ is analogous):
\begin{lstlisting}[mathescape]
let rec $\theta_1^{\IO}$ (c : $\Mone{}$ A) : $\W^{\IO}$ A = match c with
    | Ret x -> $\ret^{\W^{\IO}}$ x
    | Input k -> $\bind^{\W^{\IO}}$ (fun $\varphi$ ($h_1$, $h_2$) . forall i, $\varphi$ i (Inl i :: $h_1$, $h_2$)) (fun i . $\theta_1^{\IO}$ (k i))
    | Output o k -> $\bind^{\W^{\IO}}$ (fun $\varphi$ ($h_1$, $h_2$) . $\varphi$ () (Inr o :: $h_1$, $h_2$)) (fun () . $\theta_1^{\IO}$ k)
\end{lstlisting}
The relational rules we get by applying
our recipe to $\ioread$ and $\iowrite$ are the following:
\begin{mathpar}
  \inferrule*[left=InputL]{ }{\srelsyn{\ioread}{\ret\,a_2}{\abs{\varphi, (h_1, h_2)}{\forall i_1 \in I_1, \varphi~(i_1, a_2)~(\inl~i_1 :: h_1, h_2)}}} \and
  \inferrule*[left=OutputL]{ }{\srelsyn{\iowrite~o_1}{\ret\,a_2}{\abs{\varphi, (h_1, h_2)}{\varphi~((), a_2)~(\inr~o :: h_1, h_2)}}}
\end{mathpar}

\clearpage
\paragraph{Unbounded iteration}
Specifications for imperative programs as modeled by the $\Imp{}$ monad from
\autoref{sec:monads} come in two flavors.
This is reflected here by two unary effect observations:
a first one for total correctness $\theta^\Tot$ {\em ensuring} the termination of a
program; and a second one for partial correctness $\theta^\Part$
{\em assuming} the termination of a program.
We explain how this situation extends to the relational setting, focusing on
partial correctness, but the same methodology applies to total correctness.
Concretely, we define a simple strict relational effect observation 
\begin{align*}
\theta^\Part_{\relmarker{}} : \Imp~A_1 \times \Imp~A_2 \to \Wrel{\St}(A_1,A_2)
\end{align*}
by applying \autoref{thm:effect-commutation} to a unary  effect observation
$\theta^\Part$ defined using the domain structure with which $\W^{\St}$ 
is naturally endowed. From basic domain theoretic results,\km{reference necessary}\ch{+1}
$\W^{\St}$ can be endowed with a least fixpoint combinator
$\texttt{fix} : (\W^{\St}\,\bool \to \W^{\St}\,\bool) \to \W^{\St}\,\bool$, used
to define
\begin{lstlisting}[mathescape]
let $\theta^\Part$ (c : Imp A) : $\W^{\St}$ A = match c with
    | Ret x -> $\ret^{\W^{\St}}$ x $\hspace{1.3cm}$| Get k -> fun $\varphi$ s . $\theta^\Part$ (k s) $\varphi$ s $\hspace{1.3cm}$| Put s' k -> fun $\varphi$ s . $\theta^\Part$ k $\varphi$ s'
    | DoWhile body k -> 
      let loop (w : $\W^{\St}$ bool) = $\bind^{\W^{\St}}$ ($\theta^\Part$ body) (fun b. if b then w else $\ret^{\W^{\St}}$ $\false$) in 
      $\bind^{\W^{\St}}$ (fix loop) (fun undersc \,. $\theta^\Part$ k)
\end{lstlisting}
How does $\theta^\Part$ work? In the first three cases, it trivially
returns in the \lstinline!Ret! branch,
evaluates a continuation to the current state in the \lstinline!Get!
branch, and evaluates a continuation with an updated state in the \lstinline!Put!
branch.
The interesting part is in the \lstinline!DoWhile! branch, where the
\texttt{body} is repeatedly run using \texttt{fix} as long as the
guard returns $\true$.
We proved by induction on $c$ that $\theta^\Part$ is a monad morphism. 
\autoref{thm:effect-commutation} asks for
two monad morphisms whose images commute. We provided those morphisms
by tweaking a bit the definition of $\theta^\Part$: we embed in\ch{embed what into what?}
a variation of $\W^\St$ that accounts for a pair of states, the first
$\theta^\Part_1: \Imp \to \W^{\St}$ uses the left state and
the second $\theta^\Part_2 : \Imp \to \W^{\St}$ uses the
right state. Applying \autoref{thm:effect-commutation}, we obtain the
definition of $\theta^\Part_{\relmarker{}}$.

This simple relational effect observation
$\theta^\Part_{\relmarker{}}$ captures partial correctness in the
following sense: intuitively,
$\srelsem{\{\>\psi\>\}~c_1}{c_2}{\varphi}{\theta^\Part_{\relmarker{}}}$
implies that if $\psi(s_1,s_2)$ holds and the \emph{two} programs
$c_1$ and $c_2$ terminate on these initial states $s_1,s_2$, then the
postcondition holds of the final states. This judgment using
pre-/postconditions is expressed in terms of the usual judgment by
applying the translation into stateful backward predicate transformers
(see \autoref{sec:simple-specs}). On top of this
$\theta^\Part_{\relmarker{}}$, we devise a rule for \dowhile{} using
an invariant $\invariant_{b_1,b_2}:S \times S \to \prop$:
\newcommand\bbody{\texttt{body}}
\begin{align}
  \label{eq:simple-rules-dowhile}
  \infer{\srelsyn{\{\>\invariant_{\true,\true}\>\}~\bbody_1}{\bbody_2}{\abs{(\_,b_1,s_1)\,(\_,b_2,s_2)}{b_1 = b_2 \wedge \invariant_{b_1,b_2}(s_1,s_2)}}}
  {\srelsyn{\{\>\invariant_{\true,\true}\>\}~\dowhile\,\bbody_1}{\dowhile\,\bbody_2}{\abs{(\_,(),s_1)\,(\_,(),s_2)}{\invariant_{\false{},\false{}}(s_1,s_2)}}}
\end{align}
This rule is synchronous in the sense that the bodies always yield the same boolean
values. Consequently the two loops run the same number of steps. The
postcondition ensures that if the loop terminates, then the invariant
$\invariant_{\false,\false}$ holds.

\ch{Is a symmetric while rule enough? Does it match what's happening in RHT?}

\paragraph{Probabilistic computations} For discrete probabilistic computations
modeled by the monad $\Prob$, a first idea would be to use a unary effect
observation and appeal once again to \autoref{thm:effect-commutation}. This
simple strict relational effect observation however does not validate a rule
correlating two $\flip$ operations with an arbitrary coupling between the
Bernoulli distributions on each side~\cite{BartheGB09}. A posteriori this is not so surprising,
since the commutation hypothesis of \autoref{thm:effect-commutation} implies
that the effects on each side are observed in an independent fashion.

Hence we rely on a more sophisticated \emph{lax} relational effect observation
$\theta^\Prob : \Prob\,A_1 \times \Prob\,A_2 \to \W^\Prob(A_1,A_2)$ defined as
\[\theta^\Prob(c_1,c_2) = \abs{\varphi}{\inf_{d \sim c_1,c_2}
    \sum_{a_1:A_1,a_2:A_2} d(a_1,a_2) \cdot \varphi(a_1,a_2)}\] where we write
$d \sim c_1,c_2$ to specify a \emph{coupling} $d$ of the two distributions $c_1$
and $c_2$, \IE a distribution on $A_1 \times A_2$ such that the marginals
satisfy $\Prob(\pi_1)d = c_1$ and $\Prob(\pi_2)d = c_2$. Since we are taking the infimum
over all such couplings, the resulting relational effect observation is
necessarily lax and the conditions of linearity and continuity imposed on
$\W^\Prob$ are needed to show the monadic inequalities.\ch{you mean the lax monad morphism laws?}
Using this relational effect observation we
straightforwardly validate the following rule for correlating two
sampling operations as in ($\times$)pRHL~\cite{BartheGB09,BartheGHS17}.
\begin{align*}
  \infer{ d \sim \mathcal{B}_p, \mathcal{B}_q}{\srelsyn{\flip\,p}{\flip\,q}{\abs{\varphi}{\textstyle \sum_{b_1,b_2 : \bool} d(b_1,b_2)\cdot \varphi(b_1,b_2)}}}
\end{align*}

\subsection{Example: Noninterference}

As a specific example of the simplified framework, we explore
\emph{noninterference}, a popular relational property for information flow
control systems \cite{NanevskiBG13, clarkson10hyp, SabelfeldM03,
  AntonopoulosGHKTW17, BartheEGGKM19, BanerjeeNN16}.
The noninterference property
 dictates that the public outputs of a program cannot
depend on its private inputs. Formally, and in its most basic form, we
can capture this property by classifying the store's locations by two
security levels: \emph{high} for private information
and \emph{low} for public information.
By $s =_L s'$ we express that the
two stores $s$ and $s'$ are equal for all low locations.
We use $s \stackrel{p}{\leadsto}
s'$ to denote that the execution of a program $p$ on a store $s$ ends
in store $s'$.
The noninterference property is then written as
\[
\forall s_i, s_i', s_o, s_o'.\quad s_i =_L s_i' \wedge s_i \stackrel{p}{\leadsto} s_o \wedge s_i' \stackrel{p}{\leadsto} s_o' \implies s_o =_L s_o'
\]
A typical solution for enforcing noninterference is to define a static
type system which is capable of rejecting obviously interferent
programs \cite{SabelfeldM03}.
For example, such a type system can rule out interferent programs such as
\begin{center}
\ls$if h > 0 then l := 1 else l := 0$
\end{center}
where \verb#h# is a high reference and \verb#l# is a low one. However,
the static nature of these type systems restricts the family of programs
that we can show noninterferent. A characteristic example of this
limitation is the following noninterferent program:
\begin{center}
  \ls$if h = 1 then l := h else l := 1$
\end{center}
Relational program logics such as \citepos{benton04relational} \RHL{} provide a
less restrictive framework for proving non-interference, as the proof can rely
on information accumulated during the derivation steps.
We follow the approach of relational program logics and show how
noninterference proofs can be done in our framework.
\ch{what's this example about? this should be said upfront,
  beyond ``it's about noninterference'', which is too vague and doesn't
  really relate to our framework}%
We restrict ourselves to programs with conditionals
but without while-loops. In \autoref{sec:RHL}, we
will show a complete embedding of \RHL{},
including iteration.
For now though, we assume that we are working
with a memory consisting of locations $\mathcal{L}
= \{ \kwd{l}, \kwd{h} \}$ storing natural numbers, and consider the
data in $\kwd{h}$ to be private and the data in $\kwd{l}$ to be
public. As discussed in \autoref{sec:monads}, these stateful computations can be
captured using the monad $\St_S$ where $S
= \mathcal{L} \to \mathbb{N}$.
The program above can be represented using this monad as follows:\ch{Fonts
  for l and h not consistent}
\begin{center}
  \ls$c = let x = get h in if x = 1 then put l x else put l 1 : St one$
\end{center}
%
We instantiate our framework with the computational monad $\St_S$
on both sides, and use the simple relational specification monad
$\Wrel{\St}$ from~\autoref{sec:intro}.
The judgment we establish to prove noninterference is
\begin{mathpar}
\srelsyn{c}
        {c}
        {\abs{\varphi~(s^i_1,s^i_2)}{s^i_1~\kwd{l} = s^i_2~\kwd{l} \wedge \forall~s^f_1~s^f_2 . s^f_1~\kwd{l} = s^f_2~\kwd{l} \implies \varphi~(((),s^f_1),((),s^f_2)}}{}
\end{mathpar}
This weakest precondition transformer comes from taking the pre-/postcondition pair
\begin{align*}
\abs{(s_1,s_2)&}{s_1\,\kwd{l} = s_2\,\kwd{l}} : S {\times} S \to \prop &
\abs{(s^i_1,(),s^f_1)~(s^i_2,(),s^f_2)&}{s^f_1\,\kwd{l} = s^f_2\,\kwd{l}} : (S {\times} \one {\times} S) {\times} (S {\times} \one {\times} S) \to \prop
\end{align*}
and translating it to its predicate transformer form following the
description in~\autoref{sec:simple-specs}. The proof derivation
consists of applying the $\textsc{Bind}$ rule after a weakening, and
later applying the asymmetric conditional rules
(see page~\pageref{eq:simple-rule-if-left})
for covering the four cases.

A similar example of noninterference can be done by changing the state
monad $\St$ for the input-output monad $\IO$ described
in~\autoref{sec:monads}. In this case, input and output channels are
classified as high or low, and the noninterference policy is spelled
out in terms of these by using the specification monad $\Wrel{\IO}$ or
one of its variants.

Finally, an interesting characteristic of our framework is that we can easily adapt
the setting to handle more than one effect at the same time. 
For example, if we are interested in modeling both IO and state with
noninterference, then it is enough to apply the state monad transformer to the
IO monad, and replace the relational specification monad
$\Wrel{\St}$ by a monad which takes into account the input-output in
the specifications as well.


\section{Generic Framework}
\label{sec:generic}

While the simple framework works well for a variety of effects, it falls short of
providing a convincing treatment of control effects such as exceptions.
This limitation is due to the fact that simple relational specification monads
merge tightly together the specification of two independent computations.
We now explain how to overcome these limitations starting with the example of
exceptions, and how it leads to working inside a relational dependent type
theory.
Informed by the generic constructions on relative monads underlying
the simple setting, we derive a notion of relational specification
monad and relational effect observation in this enriched setting.
%
These relational specification monads require an important amount of
operations so we introduce relational specification monad transformers
for state and exceptions, simplifying the task of building complex
relational specification monad from simpler ones. As a consequence,
we can easily combine exceptions with any of the effects
already handled by the simplified framework
(\EG state, nondeterminism, IO, and probabilities).

\subsection{Exceptional Control Flow in Relational Reasoning}

We explained in \autoref{sec:effect-specific-rules} how to prove relational
properties of programs raising exceptions, as long as we give up on the
knowledge of which program raised an exception at the level of
relational specifications. This restriction prevents us from even stating
natural specifications such as simulations: ``if the left program raises,
so does the right one''.

In order to go beyond this unsatisfying state of affairs, we consider a type of
relational specifications allowing to write specifications
consisting of predicate transformers mapping a postcondition on pairs of
either a value or an exceptional final state to a proposition:
\[\Wrel{\Exc}(A_1,A_2) = ((A_1+E_1)\times(A_2+E_2) \to \prop) \to \prop.\]
For instance,
the specification of simulation above can be stated as
\begin{align*}
  \abs{\varphi}{\forall ae_1 ae_2. (\inr{}?\>ae_1 \Rightarrow \inr{}?\>ae_2) \Rightarrow \varphi(ae_1,ae_2)}:
  \Wrel{\Exc}(A_1,A_2),
\end{align*}
where \lstinline!Inr? ae = match ae with Inr _ -> True | _ -> False!.

As explained in the introduction, this type does not admit a monadic
operation $\bind\,w^m\,w^f$ using only a continuation of type $w^f :
A_1\times A_2 \to \Wrel{\Exc}(B_1,B_2)$ due to the fact that $w^m$
could result in an intermediate pair consisting of a normal value on
one side and an exception on the other side. Our solution is to
provide to $\bind^\Wrel{\Exc}$ the missing information it needs
in such cases. To that purpose, we use the unary
specification monads $\Wone{\Exc}A_1 = (A_1+E_1 {\to} \prop) {\to}
\prop$ and $\Wtwo{\Exc}A_2 = (A_2+E_2{\to}\prop){\to}\prop$ to provide
independent specifications of each program. With the addition of
these, we can write a combinator that relies on the unary specifications
when the results of the first computations differ (one raise an
exception and the other returns).
\begin{lstlisting}[mathescape]
val bindWrelExc : WrelExc (A_1,A_2) -> (A_1 -> WExc1 B_1) -> (A_2 -> WExc2 B_2) -> 
                                   (A_1 \times A_2 -> WrelExc (B_1,B_2)) -> WrelExc (B_1,B_2)
let bindWrelExc \,wm (f_1 : A_1 -> ((B_1 + E_1) -> Prop) -> Prop) (f_2 : A_2 -> ((B_2 + E_2) -> Prop) -> Prop) f =
  fun ($\varphi$  : (B_1 + E_1) -> Prop).
    wm (fun ae : (A_1 + E_1) \times\, (A_2 + E_2).
        match ae with
        | Inl a1, Inl a2 -> f a1 a2 $\varphi$ $\hspace{1.95cm}$| Inl a_1, Inr e_2 -> f_1 a_1 (fun be -> $\varphi$ be (Inr e_2))
        | Inr e1, Inr e2 -> $\varphi$ (Inr e1, Inr e2) $\hspace{0.9cm}$ | Inr e_1, Inl a_2 -> f_2 a_2 (fun be -> $\varphi$ (Inr e_1) be))  
\end{lstlisting}

\km{Should provide the example of a fast but incomplete check of a property
  followed by a complete one, raising when the check fails, and trying
  to prove NI}\er{how would one prove at this point without rules for
  this bind?}\km{good point...}

\subsection{A Problem of Context}
In order to keep track of these unary specifications drawn from $\Wone{\Exc}$
and $\Wtwo{\Exc}$ in the relational proofs,
we extend the relational judgment to
\[\crelsyn{}{c_1}{w_1}{c_2}{w_2}{w_{\relmarker{}}}.\]
Here, $w_1 : \Wone\Exc\,A_1$ is a unary specification for $c_1 : \Exc_{1}\,A_1$,
symmetrically $w_2 : \Wtwo\Exc A_2$ is a unary specification for
$c_2 : \Exc_{2}\,A_2$, and $w_{\relmarker{}} : \Wrel\Exc(A_1,A_2)$ specifies the
relation between the programs $c_1$ and $c_2$.
%
Using this richer judgment, we would like a rule for sequencing computations as
follows, where a bold variable $\boldsymbol{w}$ stands for the triple
$(w_1,w_2,w_{\relmarker{}})$:
\[
  \infer{\crelsyn{}{m_1}{w^m_1}{m_2}{w^m_2}{w^m_{\relmarker}}\\
    \forall a_1,a_2\crelsyn{}{f_1\;a_1}{w^f_1\;a_1}{f_2\;a_2}{w^m_2\;a_2}{w^f_{\relmarker{}}\;a_1\;a_2}
  }{\crelsyn{}{\bind^{\Exc_1}\;m_1\;f_1}{\bind^{\Wone\Exc}\;w^m_1\;w^f_1}{\bind^{\Exc_2}\;m_2\;f_2}{\bind^{\Wtwo\Exc}\;w^m_2\;w^f_2}
    {\bind^{\Wrel\Exc}\;\boldsymbol{w^m}\;\boldsymbol{w^f}}}
\]

What would the semantics of such a relational judgment be?
A reasonable answer at first sight is to state formally the previous intuition
in terms of unary and relational effect observations:\ch{Why not add the effect
  observations $\theta^\Exc_1,\theta^\Exc_2,\theta^\Exc$ on the sem symbol?}
\[\crelsem{}{c_1}{w_1}{c_2}{w_2}{w_{\relmarker{}}}\quad =\quad
  \theta^\Exc_1\,c_1 \leq w_1 \enspace\wedge\enspace \theta^\Exc_2\,c_2 \leq w_2
  \enspace\wedge\enspace \theta^\Exc_{\relmarker}(c_1,c_2) \leq w_{\relmarker{}}
\]
However this naive attempt does not validate the rule for sequential
composition above. The problem lies in the management of context. To
prove the soundness of this rule, we have in particular to show that
$\theta^\Exc_1\,(\bind^\Exc_1\,m_1\,f_1)\leq
\bind^{\Wone\Exc}\,w^m_1\,w^f_1$ under the hypothesis
$\theta^\Exc_1\,m_1 \leq w^m_1\wedge \ldots$ and $\forall a_1, a_2,
\theta^{\Wone\Exc}\,(f_1\,a_1) \leq w^f_1\, a_1 \wedge \ldots$, in
particular the second hypothesis requires an element $a_2 : A_2$ that
prevents\footnote{Instead of insisting that
  $\crelsyn{}{c_1}{w_1}{c_2}{w_2}{w_{\relmarker{}}}$ proves the correctness of $c_1$
  and $c_2$ with respect to $w_1$ and $w_2$ we could try to presuppose
  it, however this idea does not fare well since it would require a
  property akin of cancellability with respect to bind
  $\theta^\Exc_1\,(\bind^\Exc_1\,m_1\,f_1)\leq
  \bind^{\Wone\Exc}\,w^m_1\,w^f_1 \Rightarrow \theta^\Exc_1\,m_1 \leq
  w^m_1$ that has no reason to hold in our examples.} us from
concluding by monotonicity of $\bind^\Wone\Exc$.

This problematic hypothesis\er{the second one?} only depends on the part of
the context relevant for the left program and not on the full context,
so we introduce structured contexts $\Gamma =
(\Gamma_1,\Gamma_2)$ in our judgments, where $\Gamma_1$ and
$\Gamma_2$ are simple contexts. The judgment
$\crelsyn{\Gamma}{c_1}{w_1}{c_2}{w_2}{w_{\relmarker}}$ now presupposes that $\Gamma_i \vdash
c_i : \M_i\,A_i$, $\Gamma_i \vdash w_i : \W_i$ ($i = 1,2$) and that
$\Gamma_1,\Gamma_2\vdash w_{\relmarker} : \Wrel{}(A_1,A_2)$.   
The semantics of this judgment is given by
\begin{align}
  \label{eq:crelsem-def}
  \crelsem{\Gamma}{c_1}{w_1}{c_2}{w_2}{w_{\relmarker{}}} \quad=\quad
  \left (
    \begin{array}{c}
      \forall \gamma_1 : \Gamma_1, \theta_1(c_1\,\gamma_1) \leq w_1\,\gamma_1,\\
      \forall \gamma_2 : \Gamma_2, \theta_2(c_2\,\gamma_2) \leq w_2\,\gamma_2,\\
      \forall (\gamma_1, \gamma_2) : \Gamma_1 \times \Gamma_2,
      \theta_{\relmarker}(c_1\,\gamma_1, c_2\,\gamma_2) \leq
      w_{\relmarker{}}(\gamma_1, \gamma_2)
    \end{array}
  \right )
\end{align}
A conceptual understanding of this interpretation that will be useful in the
following is to consider $\Gamma$ as a (trivial) relation $\br\Gamma =
(\Gamma_1,\Gamma_2, \abs{(\gamma_1:\Gamma_1)(\gamma_2:\Gamma_2)}\one)$ instead
of a pair and define the family of relations $\br\Theta(\boldsymbol{\gamma}) = (\Theta_1(\gamma_1), \Theta_2(\gamma_2),\Theta_{\relmarker{}}\boldsymbol{\gamma})$ dependent over $\br\Gamma$:
\begin{mathpar}
      \Theta_1(\gamma_1 : \Gamma_1) = \theta_1(c_1\,\gamma_1) \leq w_1\,\gamma_1,\and
      \Theta_2(\gamma_2 : \Gamma_2) = \theta_2(c_2\,\gamma_2) \leq w_2\,\gamma_2,\and
      \Theta_{\relmarker{}}(\boldsymbol{\gamma}:\Gamma)(w_1:\Theta_1\,\gamma_1, w_2:\Theta_2\,\gamma_2) = \theta_{\relmarker}(c_1\,\gamma_1, c_2\,\gamma_2) \leq w_{\relmarker{}}\boldsymbol{\gamma}.
\end{mathpar}
Then the relational judgment
$\crelsem{\Gamma}{c_1}{w_1}{c_2}{w_2}{w_{\relmarker{}}}$ can be
interpreted as a dependent function $\left(\boldsymbol{\gamma}: \br\Gamma\right) \to
\br\Theta\,\boldsymbol{\gamma}$ in an appropriate relational dependent type theory.

\subsection{A Relational Dependent Type Theory}
\label{sec:relational-dependent-type-theory}

Adding unary specifications in the relational judgment enables a full
treatment of exceptions, however
the pure rules of section \autoref{sec:simple-pure-relational-rules}
do not deal with a structured context $\br\Gamma = (\Gamma_1,\Gamma_2, \Gamma_{\relmarker{}})$.
In order to recover rules dealing with such a context, we apply the same recipe
internally to a relational dependent type theory as described by \citet{Ton13}.
In practice, this type theory could be described as a syntactic model in the sense of
\citet{BoulierPT17}, that is a translation from a source type theory to a target
type theory that we take to be our ambient type theory, where a type in the
source theory is translated to a pair of types and a relation between them. We
call the resulting source type theory \RDTT{} and describe part\ch{?} of its
construction in \autoref{fig:syntaxRDTT}.
A systematic construction of \RDTT{} at
the semantic level is obtained by considering the category with families
$\Span{\TypeCat}$ consisting of families of types and functions
indexed by the span $\left(1 \leftarrow \relobj{} \rightarrow
  2\right)$\er{category?}\km{As a diagram, it's a span, no?}, a special case of
\citet{Shulman14,KL18}.
\begin{figure}
  \begin{mathpar}
    \br{A},\br{B},\br{\Gamma} ::= \br\zero \alt \br\one \alt \br\bool \alt \br\bN
    \alt \br{A}+\br{B} \alt (\boldsymbol{a}: \br{A}) \times \br{B}\;\boldsymbol{a} \alt
    (\boldsymbol{a}:\br{A}) \to \br{B}\;\boldsymbol{a}
    \\
    \sem{-} \text{ maps a relational type }\br{A}\text{ to its underlying
    representation } \sem{\br{A}} = (A_0,A_1,A_r)\and
    \sem{\br\zero} = (\zero{}, \zero, =)
    \and
    \sem{\br\one} = (\one{}, \one, =)
    \and
    \sem{\br\bool} = (\bool, \bool, =)
    \and
    \sem{\br\bN} = (\bN, \bN, =)
    \and
    \sem{\br{A}+\br{B}} = \left (
      \begin{array}{l}
        ab_1 : A_1 + B_1\\
        ab_2 : A_2 + B_2 \\
        \case{(ab_1,ab_2)}{(\inl\,a_1,\inl\,a_2)}{A_{\relmarker{}}\,a_1\,a_2}{(\inr\,b_1,\inr\,b_2)}{B_{\relmarker{}}\,b_1\,b_2\alt
        (\_,\_)~.~\zero}
      \end{array}
    \right )
    \and
    \sem{(\boldsymbol{a}: \br{A}) \times \br{B}\;\boldsymbol{a}} =
    \left (\begin{array}{l}
    (a_1, b_1) : (a_1 : A_1) \times B_1\;a_1,\\
    (a_2,b_2) : (a_2 : A_2) \times B_2\;a_2,\\
    (a_r : A_r\;a_1\;a_2) \times B_r\, a_1\;a_2\;a_r\; b_1\;b_2
    \end{array}\right )
  \and
    \sem{(\boldsymbol{a}: \br{A}) \to \br{B}\;\boldsymbol{a}} =
    \left (\begin{array}{l}
    f_1 : (a_1 : A_1) \to B_1\;a_1,\\
    f_2 : (a_2 : A_2) \to B_2\;a_2,\\
    (a_1 : A_1)(a_2:A_2)(a_r : A_r\;a_1\;a_2) \to B_r\, a_1\;a_2\;a_r\; (f_1\;a_1)\;(f_2\;a_2)
    \end{array}\right )
  \end{mathpar}
  \caption{Syntax of \RDTT{} and translation to base type theory}
  \label{fig:syntaxRDTT}
\end{figure}

Moving from our ambient type theory to \RDTT{} informs us on how to define rules
coming from the type theory.
For instance, generalizing the rule for if-then-else, we can use the
motive $P(\boldsymbol{ab}:\br{A} + \br{B}) = \br\Theta(\boldsymbol{ab}) :
\br\Type$ on the dependent eliminator for sum type
\begin{figure}
  \begin{align*}
    \sem{\texttt{elim\_sum}}:~&(P_1 : A_1 + B_1 \to \Type) \to (P_2 : A_2 + B_2 \to \Type) \to\\
                              &(P_{\relmarker{}} : \forall (ab_1 :A_1 + B_1) (ab_2:A_2+B_2), (\br{A}+\br{B})_{\relmarker{}}\,ab_1\,ab_2 \to \Type) \to\\
                              &(\forall (a_1:A_1), P_1\,(\inl\,a_1)) \to (\forall (a_2:A_2), P_2\,(\inl\,a_2)) \to \\
                              &(\forall a_1\, a_2\,(a_{\relmarker{}} : \br{A}\,a_1\,a_2), P_{\relmarker{}}\,(\inl\,a_1)\,(\inl\,a_2)\,a_{\relmarker{}})\\
                              &(\forall (b_1:B_1), P_1\,(\inr\,b_1)) \to (\forall (b_2:B_2), P_2\,(\inr\,b_2)) \to \\
                              &(\forall b_1\, b_2\,(b_{\relmarker{}} : \br{B}\,b_1\,b_2), P_{\relmarker{}}\,(\inr\,b_1)\,(\inr\,b_2)\,b_{\relmarker{}}) \to\\
                              &\forall ab_1\,ab_2\,(ab_{\relmarker{}}: (\br{A}+\br{B})_{\relmarker{}}\,ab_1\,ab_2), P_{\relmarker{}}\,ab_1\,ab_2\,ab_{\relmarker{}}
  \end{align*}
  \caption{Relational translation of the eliminator for sum types}
  \label{fig:rel-trans-sum-type}
\end{figure}
\[\texttt{elim\_sum} : (P :(\br{A} + \br{B}) {\to} \br\Type) \to (a:\br{A} {\to}
  P\,a) \to (b:\br{B} {\to} P\,b) \to (x: \br{A}+\br{B}) {\to} P\,x\]
to obtain a rule for case splitting. This eliminator translates to the large
term in~\autoref{fig:rel-trans-sum-type} ~that induces the following relational rule using
$\boldsymbol{w^l} = (w^l_1, w^l_2, w^l_{\relmarker{}})$, 
$\boldsymbol{w^r} = (w^r_1, w^r_2, w^r_{\relmarker{}})$
and the relational specifications of the conclusion -- where we abbreviate
pattern matching with a $\texttt{case}$ construction -- as arguments to the eliminator
\[
  \infer{
    \crelsyn{\Gamma,\boldsymbol{a} : \br{A}} {c_1[\inl\;a_1/ab_1]}{w^l_1} {c_2[\inl\;a_2/ab_2]}{w^l_2} {w^l_{\relmarker{}}[a_{\relmarker{}}/ab_{\relmarker{}}]}\\
    \crelsyn{\Gamma,\boldsymbol{b} : \br{B}} {c_1[\inr\;b_1/ab_1]}{w^r_1} {c_2[\inr\;b_2/ab_2]}{w^r_2} {w^r_{\relmarker{}}[b_{\relmarker{}}/ab_{\relmarker{}}]}
  }
  {\crelsynL{\br\Gamma,\boldsymbol{ab} : \br{A}{+}\br{B}}
    {c_1}{\case{ab_1}{\inl\;a_1}{w^l_1}{\inr\;b_1}{w^r_1}}
    {c_2}{\case{ab_2}{\inl\;a_2}{w^l_2}{\inr\;b_2}{w^r_2}}
    {\caseL{ab_1,ab_2}{\inl\;a_1,\inl\;a_2}{w^l_{\relmarker{}}}{\inr\;b_1,\inr\;b_2}{w^r_{\relmarker{}}}}}
\]
As in the simple setting, we can then refine this rule to obtain synchronous or
asynchronous rules specifying a required shape for the programs $c_1,c_2$.

\subsection{Relative Monads and Monad Morphisms}
\label{sec:RelMon}
Before giving the general framework able to derive monadic rules
dealing with exceptions, we return to the notions of relative monads
and relative monad morphisms, since these will be the common
underlying concept relating the simple and generic frameworks.

\begin{definition}[relative monads \cite{AltenkirchCU14}]
  \label{def:RelMon}
  Let $\I, \C$ be categories and ${\J : \I {\to} \C}$ a functor between
  these. A $\J$-relative monad is given by 
  \begin{itemize}
  \item for each $A\in\I$, an object $\T\,A\in\C$
  \item for each $A\in\I$, a morphism $\ret^\T_A \in \C(\J\,A;\T\,A)$
  \item for each $A,B\in\I$, a function $\revBindT{(-)} : \C(\J\,A;\T\,B) \to \C(\T\,A;\T\,B)$
  \item satisfying the 3 following equations
    \begin{mathpar}
      \revBindT{f} \circ \ret^\T_A = f
      \and
      \revBindT{(\ret^\T_A)} = \id_{\T\,A}
      \and
      \revBindT{g} \circ \revBindT{f} = \revBindT{(\revBindT{g}\circ f)}
    \end{mathpar}
  \end{itemize}
\end{definition}
Noting $\TypeCat$ for the category of types and functions of our ambient type
theory, $\Poset$ for the category of preordered sets and monotonic
functions, and $\Discr : \TypeCat \to \Poset$ the functor equipping a type with
its discrete preorder structure , a simple relational specification monad could be described as a
relative monad $\Wrel{} : \TypeCat^2 {\to} \Poset$ over the functor $\Discr
\circ \times : \TypeCat^2 {\to} \Poset$ sending a pair of types $(A_1,A_2)$
to their product $A_1 \times A_2$ equipped with a discrete preorder.
The monotonicity condition imposed on $\bind^\Wrel{}$ amounts to require that
all the structure is enriched in $\Poset$\footnote{$\TypeCat^2$ can be enriched
  over $\Poset$ by change-of-enrichment through the monoidal functor $\Discr$.}
\cite{kelly1982basic}.

Simple relational effect observations from $\Mone{},\Mtwo{}$ to $\Wrel{}$ can
also be interpreted as instances of relative monad morphisms. First, a pair of
computational monads $\Mone,\Mtwo$ yields a monad $\Mone{}\otimes\Mtwo{} :
\TypeCat^2 \to \TypeCat^2$ acting on pairs of types, that is
$\Mone\otimes\Mtwo\,(A_1,A_2) = (\Mone\,A_1, \Mtwo\,A_2)$ with monadic
structures provided by each sides. Second, by proposition 2.3 of
\citet{AltenkirchCU14}, the monad $\Mone{}\otimes\Mtwo{}$ is a relative monad on
the identity functor $\Id_{\TypeCat^2}$. A simple relational effect observation
is a relative monad morphism from $\Mone{}\otimes\Mtwo$ to $\Wrel{}$ over the
functor $\Discr\circ \dottimes$.
\begin{definition}[Lax relative monad morphism]
  \label{def:relMonMorph}
  Let $\I,\C_1,\C_2$ be categories enriched over $\OrdCat$ and $\J_1 : \I \to
  \C_1, \J_2 : \I \to \C_2, \mathcal{F} : \C_1 \to \C_2$ be $\OrdCat$-enriched
  functors such that $\varphi : \mathcal{F} \circ \J_1 \cong \J_2$. A {\em
    lax relative monad morphism} from a $\J_1$-relative monad $\T_1 : \I \to \C_1$
  to a $\J_2$-relative monad $\T_2 : \I \to \C_2$ is
  \begin{itemize}
  \item a family of morphisms $\theta_A : \mathcal{F}\circ \T_1\,A \to \T_2\,A$
    indexed by objects $A \in \mathcal{I}$,
  \item such that
    \begin{align}
      \label{eq:relative-monad-morphism-ineq}
      \theta \circ \mathcal{F}\,\ret^{\T_1} \leq \ret^{\T_2} \circ \varphi
      \qquad&\qquad
      \theta \circ \mathcal{F}(\revBind{f}{{\T_1}}) \leq \revBind{(\theta \circ
        \mathcal{F}f \circ \varphi^{-1})}{{\T_2}} \circ \theta
    \end{align}
    We say that $\theta$ is a \emph{relative monad morphism} when the last two
    conditions are equalities.
  \end{itemize}
\end{definition}

Crucially, this definition of relative monad morphism generalizes the notion
defined by \citet{AltenkirchCU14} by enabling different base functor, a relative
monad analog to the monad opfunctors of \citet{Street72}.\footnote{However, in
  contrast to the situation of monads, these relative monad morphisms do not
  dualize well.} Up to the enrichment,
we recover the definition of \citet{AltenkirchCU14} by taking $\J_1 = \J_2$,
$\mathcal{F} = \Id$ and $\varphi = \id$.


\subsection{Relational Specification Monads, Relational Effect Observations}
\label{sec:rel-spec-monad-rel-effect-obs}

Motivated by the case of exceptions, we now define the general notion of a
relational specification monad. This definition is obtained by instantiating the
definitions of an (enriched) relative monad to our relational dependent type
theory, ensuring that we obtain a theory uniform with the simple setting, and
crucially that we can use a similar methodology to introduce relational rules.
We note $\Span{\OrdCat}$ for the category of relations between ordered types,
$\Jprod : \Type \times \Type \to \Span{\OrdCat}$ for the functor defined on
objects by $\Jprod(A_1,A_2) = \Discr\,A_1 {\xleftarrow{\pi_1}} \Discr\,A_1{\times}
\Discr\,A_2 {\xrightarrow{\pi_2}} \Discr\,A_2$ and $\pi_{1,2} : \Span{\OrdCat}
\to \Type \times \Type$ sending a relation
$A_1{\leftarrow}A_{\relmarker{}}{\rightarrow}A_2$ to its legs $(A_1,A_2)$.
\begin{definition}
  \label{def:cRelSpecMon}
  A {\em relational specification monad} consist of
  a pair of unary specification monads $\Wone{}, \Wtwo{} : \Type \to \OrdCat$
  and a relative monad $\boldsymbol{\W{}} : \Type \times \Type \to
  \Span{\OrdCat}$ over $\Jprod$ lifting $\Wone{}, \Wtwo{}$, that is such that
  $\pi_{1,2} \circ \boldsymbol{\W{}} = \Wone{}\times \Wtwo{}$, and
  whenever $\pi_{1,2}(f) = (f_1,f_2)$
  \begin{mathpar}
    \pi_{1,2}(\ret^{\boldsymbol{\W{}}}) = (\ret^\Wone{},\ret^\Wtwo{}),
    \and
    \pi_{1,2}(\bind^{\boldsymbol{\W{}}}\,f) =
    (\bind^\Wone{}\,f_1,\bind^\Wtwo{}\,f_2).
  \end{mathpar}
\end{definition}

In components, a relational specification monad over unary specification monads $\Wone{}, \Wtwo{}$ consists
of a relation $\Wrel{}(A_1,A_2) : \Wone{}A_1 \to \Wtwo{}A_2 \to
\Type$ equipped with a preorder $\leq^\W$, and operations 
\begin{align*}
  \ret^\Wrel{} : (a_1,a_2)&:A_1 {\times} A_2 \to \Wrel{}(A_1,A_2)\,(\ret^\Wone{}\,a_1)\,(\ret^\Wtwo{}\,a_2)\\
  \bind^\Wrel{} :\quad\>\, w^m_1&:\Wone{}A_1 \to w^m_2:\Wtwo{}A_2 \to w^m_{\relmarker{}}:\Wrel{}(A_1,A_2)\,w^m_1\,w^m_2 \to\\
  w^f_1&:(A_1 \to \Wone{}B_1) \to w^f_2: (A_2 \to \Wtwo{}B_1) \to \\
  w^f_{\relmarker{}}&:(((a_1,a_2):A_1\times A_2) \to \Wrel{}(B_1,B_2)\,(w^f_1\,a_1)\,(w^f_2\,a_2)) \to\\
  \Wrel{}&(B_1,B_2)\,(\bind^\Wone{}w^m_1\,w^f_1)\,(\bind^\Wtwo{}w^m_2\,w^f_2) 
\end{align*}
satisfying equations analogous to the monadic laws.

If these operations look complex, in most of our examples the relation
$\Wrel{}(A_1,A_2)\,w_1\,w_2$ is independent of $w_1$ and $w_2$.
This happens for our leading example of exceptions, but also for any relational
specification monad constructed out of a simple relational specification monad.
Indeed, we can associate to any simple relational specification monad $\Wrel{}$
the relational specification monad $\W(A_1,A_2) =
(\Wrel{}(A_1,\one),\Wrel{}(\one,A_2), \abs{w_1\,w_2}{\Wrel{}(A_1,A_2)})$. The
monadic operations just discard the superfluous arguments.
We now turn to the definition of a relational effect observation.
\begin{definition}[Lax relational effect observation]
  A {\em lax relational effect observation} from $M_1,M_2$ to the relational
  specification monad $\boldsymbol{\W{}}$ over $\Wone{}, \Wtwo{}$ is a lax
  relative monad morphism $\boldsymbol{\theta}$ from $M_1\otimes
  M_2$ to $\boldsymbol{\W{}}$ over $\Jprod$. A relational effect observation is
  \emph{strict} when the inequalities~(\ref{eq:relative-monad-morphism-ineq}) hold as equalities.
\end{definition}
Explicitly, such a lax relational effect observation
$\boldsymbol{\theta}$ of three components $\theta_1,\theta_2, \theta_{\relmarker{}}$
where $\theta_1 : \Mone \to \Wone{}$, $\theta_2 : \Mtwo{} \to \Wtwo{}$
are (plain) monad morphisms, and 
\[\theta_{\relmarker{}}\quad:\quad
((m_1,m_2):\Mone{}A_1\times\Mtwo{}A_2) \to \Wrel{}(A_1,A_2)\,(\theta_1\,m_1)\,(\theta_2\,m_2)\] 
verify the two inequations with respect to the monadic operations
\begin{align*}
  \theta_{\relmarker{}}(\ret^\Mone{}\,a_1,\ret^\Mtwo{}\,a_2) &\leq \ret^\Wrel{}\,(a_1,a_2)
  : \Wrel{}(A_1,A_2)\,(\theta_1\,(\ret^\Mone{}\,a_1))\,(\theta_2\,(\ret^\Mtwo{}\,a_2))\\
  \theta_{\relmarker{}}(\bind^\Mone{}m_1\,f_1,\bind^\Mtwo{}m_2\,f_2) &\leq \bind^\Wrel{}\,(\theta_1\,m_1)\,(\theta_2\,m_2) \,(\theta_{\relmarker{}}\,m_{\relmarker{}})
  ~\theta_1{\circ}f_1\enspace\theta_2{\circ}f_2\enspace\theta_{\relmarker{}}{\circ}(f_1\times f_2)
\end{align*}
Given a relational effect observation $\boldsymbol{\theta} : \Mone{}\otimes\Mtwo{}
\to \boldsymbol{\W}$, we can define in full generality the semantics of the
relational judgment by the \autoref{eq:crelsem-def}.\ch{I'm confused,
  all this fuss about relational DTT just to define this without
  using the RDTT function space, as promised at the end of 3.2?}
We introduce the
generic monadic rules in \autoref{fig:complexRelRules}, and similarly
to the simple setting obtain the following soundness theorem.
\begin{theorem}[Soundness of monadic rules]
 The relational rules in \autoref{fig:complexRelRules} are sound with
 respect to any lax relational effect observation $\boldsymbol{\theta}$, that is 
 \[\crelsyn{\br\Gamma}{c_1}{w_1}{c_2}{w_2}{w_{\relmarker{}}}\quad\Rightarrow{}\quad
 \forall \boldsymbol{\theta},\enspace \crelsemext{\br\Gamma}{c_1}{w_1}{c_2}{w_2}{w_{\relmarker{}}}{\boldsymbol{\theta}}\]
\end{theorem}

\begin{figure}
  \begin{mathpar}
    \inferrule*[left=Weaken]{\crelsyn{\br\Gamma}{c_1}{w_1}{c_2}{w_2}{w_{\relmarker{}}}\\
      w_1 \leq^\Wone{} w'_1 \\
      w_2 \leq^\Wtwo{} w'_2 \\
      w_{\relmarker{}} \leq^\Wrel{} w'_{\relmarker{}}}
    {\crelsyn{\br\Gamma}{c_1}{w'_1}{c_2}{w'_2}{w'_{\relmarker{}}}}
    \and
    \inferrule*[left=Ret]{\Gamma_1 \vdash a_1:A_1\\ \Gamma_2 \vdash a_2:A_2}
    {\crelsyn{\br\Gamma}{\ret^\Mone{}a_1}{\ret^\Wone{}a_1}{\ret^\Mtwo{}a_2}{\ret^\Wtwo{}a_2}{\ret^\Wrel{}(a_1,a_2)}}
    \and
    \inferrule*[left=Bind]{\crelsyn{\br\Gamma}{m_1}{w^m_1}{m_2}{w^m_2}{w^m}\\
      \crelsyn{\br\Gamma, \boldsymbol{a} : \br{A}}{f_1\;a_1}{w^f_1\;a_1}{f_2\;a_2}{w^m_2\;a_2}{w^f\;\boldsymbol{a}}
    }{\crelsynL{\br\Gamma}{\bind^\Mone{}\;m_1\;f_1}{\bind^\Wone{}\;w^m_1\;w^f_1}{\bind^\Mtwo{}\;m_2\;f_2}{\bind^\Wtwo{}\;w^m_2\;w^f_2}{\bind^\Wrel{}\;\boldsymbol{w^m}\;\boldsymbol{w^f}}}
  \end{mathpar}
  \caption{Generic monadic rules in the full relational setting}
  \label{fig:complexRelRules}
\end{figure}

\subsection{Relational Specification Monad Transformers}
\label{sec:relational-specification-monad-transformers}

Having a category of relational specification monads and relative monad
morphisms between them, we define a \emph{relational specification monad
  transformer} to be a pointed endofunctor on this category~\cite{LuthG02}.
Under mild assumptions, the usual state and exception transformer lifts to this
setting, yielding in each case both a left-variant and a right-variant applying
either to the left type $A_1$ or right one $A_2$ of a relational specification
monad $\W(A_1,A_2)$. Since the two variants are symmetric, we only detail the
left ones.

\paragraph{Adding state} The usual state monad transformer maps a monad $\M$ to the
monad $\StT(\M)\,A = S \to \M\,(A\times S)$. The left relational state monad transformer
$\StT_{\relmarker}$ maps a relational specification monad $\boldsymbol{\W}\,(A_1,A_2) = (\Wone{}\,A_1,
\Wtwo{}\,A_2, \abs{w_1\,w_2}{\Wrel{}(A_1,A_2)\,w_1\,w_2})$ to the relational specification monad
with carrier
\[\StT_{\relmarker{}}(\boldsymbol{\W})(A_1,A_2) \enspace =\enspace (\StT(\Wone{})\,A_1,\enspace \Wtwo{}A_2,\enspace
  \abs{w_1\,w_2}{(s_1:S_1) \to \Wrel{}(A_1\times S_1,A_2)\,(w_1\,s_1)\,w_2})\]

The monadic operations on $\StT_{\relmarker{}}(\W)_{1}$ are given by the usual
state transformer. The added data resides in the $\ret$ and $\bind$ operations
responsible for the relational part:
\begin{lstlisting}[mathescape]
let $\ret^{\StT(\W)}_{\relmarker{}}$ (a_1,a_2) : (s_1: S_1) -> W_rel (A_1\times\,S_1,A_2) ($\ret^{StT(W)_1}$ (a_1,s_1)) ($\ret^\Wtwo{}$ a_2) = fun s_1. $\ret^\Wrel{}$ ((a_1,s_1), a_2)

let $\bind^{\StT(\W)}_{\relmarker{}}$ (m_1 : $\StT(\W)$_1 A_1) (m_2 : W_2 A_2) (m_rel : $\StT(\W)$_rel (A_1,A_2) m_1 m_2)
               (f_1 : A_1 -> $\StT(\W)$_1 B_1) (f_2 : A_2 -> W_2 B_2)
               (f_rel : (a_1,a_2):A_1\times\,A_2 -> $\StT(\W)$_rel (B_1, B_2) (f_1 a_1) (f_2 a_2))
               : $\StT(\W)$_rel (B_1,B_2) ($\bind^{\StT(\W)_1}$ m_1 f_1) ($\bind^\Wtwo{}$ m_2 f_2) =
  fun s_1. $\bind^\Wrel{}$ (m_1 s_1) m_2 (m_rel s_1) (fun (a_1,s_1'). f_1 a_1 s_1') f_2 (fun ((a_1,s_1'), a_2). f_rel (a_1,a_2) s_1')
\end{lstlisting}

\paragraph{Adding exceptions}
In a similar flavor, the exception monad transformer $\ExcT$ mapping a monad
$\M$ to $\ExcT(\M)A=\M(A + E_1)$ gives raise to the carrier its relational
specification monad counterpart $\ExcT_{\relmarker{}}(\boldsymbol{\W})(A_1,A_2) =
(\ExcT(\Wone{})A_1, \Wtwo{}A_2, \Wrel{}(A_1+E_1, A_2))$.

However, in order to define the bind operation we need to restrict our attention
to relational specification monad for which unary specification can be lifted to
relational ones. This is provided by the structure of two maps:
\begin{align*}
  \tau_1 : w_1:\Wone{}(A_1,\one) \to \Wrel{}(A_1,\one)\,w_1\, (\ret^\Wtwo{}\,()),\\
  \tau_2: w_2:\Wtwo{}(\one,A_2) \to \Wrel{}(\one,A_2)\, (\ret^\Wone{}\,())\,w_2.
\end{align*}
such that pairing each of them with identity provides a
monad morphism, that is $(\id,\tau_1) :\Wone{}(A,\one) \to
(w:\Wone{}(A,\one)) \times \Wrel{}(A,\one)\,w\,(\ret^\Wtwo{}\,())$ respects the
monadic equations. Any relational specification monad induced by a simple one has a canonical such
structure, taking $\tau_1$ and $\tau_2$ to be identity. The state transformer
and the exception transformer also preserve this structure.

Assuming that $\boldsymbol{\W}$ is equipped with $\tau_1,\tau_2$, we define the
return and bind operation on $\ExcT(\Wrel{})$ as:
\begin{lstlisting}
let $\ret^{\ExcT(\W)_{\relmarker{}}}$ (a_1,a_2) : W_rel (A_1 + E_1, A_2) ($\ret^{\ExcT(\W)_{1}}$ a_1) ($\ret^{\W_{2}}$ a_2) = $\ret^{\W_{\relmarker{}}}$ (Inl a_1, a_2)

let $\bind^{\ExcT(\W)_{\relmarker{}}}$ (m_1 : $\ExcT(\W)$_1 A_1) (m_2 : W_2 A_2) (m_rel : $\ExcT(\W)$_rel (A_1,A_2) m_1 m_2)
               (f_1 : A_1 -> $\ExcT(\W)$_1 B_1) (f_2 : A_2 -> W_2 B_2) 
               (f_rel : (a_1,a_2):A_1\times\,A_2 -> $\ExcT(\W)$_rel (B_1, B_2) (f_1 a_1) (f_2 a_2))
               : $\ExcT(\W)$_rel (B_1,B_2) ($\bind^{\ExcT(\W)_{1}}$ m_1 f_1) ($\bind^{\W_{2}}$ m_2 f_2) =
  $\bind^{\W_{\relmarker{}}}$ m_1 m_2 m_rel (fun ae_1. match ae_1 with | Inl a_1 -> f_1 a_1 | Inr e_1 -> $\ret^{\W_{1}}$ (Inr e_1)) f_2
          (fun ae_1 a_2 . match ae_1 with 
                      | Inl a_1 -> f_rel a_1 a_2
                      | Inr e_1 -> $\bind^{\W_{\relmarker{}}}$ (tau_2 (f_2 a_2)) (fun ((), b_2) . $\ret^{\W_{\relmarker{}}}$ (Inr e_1, b_2)))
\end{lstlisting}

\begin{figure}
  \begin{mathpar}
    \inferrule*[left=]{
    }{\crelsyn{\br\Gamma}{\throw\,e_1}{\abs{\varphi_1}{\varphi_1\,(\inr\,e_1)}}{\ret^\Exc\,a_2}{\ret^{\Wtwo{\Exc}}a_2}{\abs{\varphi}{\varphi\,(\inr\,e_1,\inl\,a_2)}}}
    \and
    \inferrule*[left=]{ }{\crelsyn{\br\Gamma}{\ret^\Exc\,a_1}{\ret^{\Wone{\Exc}}a_1}{\throw\,e_2}{\abs{\varphi_2}{\varphi_2\,(\inr\,e_2)}}{\abs{\varphi}{\varphi\,(\inl\,a_1,\inr\,e_2)}}}
    \and
    \inferrule*[left=]{
      \crelsyn{\br\Gamma}{c_1}{w_1}{c_2}{w_2}{w_{\relmarker{}}}\\
      \crelsyn{\br\Gamma}{c^{err}_1}{w^{err}_1}{c^{err}_2}{w^{err}_2}{w^{err}_{\relmarker{}}}
    }{\crelsyn{\br\Gamma}{\catch\,c_1\,c^{err}_1}{w^\catch\,w_1\,w^{err}_1}{\catch\,c_2\,c^{err}_2}{w^\catch\,w_2\,w^{err}_2}{w^\catch_{\relmarker{}}\,
      w_{\relmarker{}}\,\boldsymbol{w^{err}}}}
  \end{mathpar}
  \vspace{0.2cm}
\begin{lstlisting}
        let $w^\catch$ (w : WExc A) (werr : E -> WExc A) : W A = 
          fun $\varphi$. w (fun ae. match ae with | Inl a -> $\ret^{\W^\Exc}$ a $\varphi$ | Inr e -> werr e $\varphi$)

        let $w^\catch_{\relmarker{}}$ (w:WrelExc (A_1,A_2)) (werr_1 : E_1 -> WExc_1 A_1) (werr_2 : E_2 -> WExc_2 A_2) 
                      (werr_rel : E_1 \times\, E_2 -> WrelExc (A_1,A_2)) : WrelExc (A_1,A_2) =
          fun $\varphi$. w (fun (ae_1, ae_2). match ae_1, ae_2 with
            | Inl a_1, Inl a_2 -> retWrelExc (a_1,a_2) $\varphi$
            | Inr e_1, Inl a_2 -> werr_1 e_1 (fun ae_1 -> $\varphi$ ($ae_1$, Inl $a_2$))
            | Inl a_1, Inr e_2 -> werr_2 e_2 (fun ae_2 -> $\varphi$ (Inl $a_1$, $ae_2$))
            | Inr e_1, Inr e_2 -> werr_rel (e_1,e_2) $\varphi$)
\end{lstlisting}
  \caption{Rules for exceptions}
  \label{fig:full-rules-exceptions}
\end{figure}

Putting these monad transformer to practice, we can finally define the full
relational specification monad for exceptions validating the rules in
\autoref{fig:full-rules-exceptions} by first lifting the simple relational
$\Wrel{\Pure}$ and applying the exception transformers on both
left and right sides. Further, applications would involve
specifications relating state and exceptions with rollback state.
\ch{TODO: Ram used an asynchronous catch rule. Show that instead?
  Can that be used to derive the synchronous rule?
  The other way the asynchronous rule can probably be derived by catching a return/raise.}
The structure provided by $\tau_1,\tau_2$ is a technical requirement, and we
leave to further investigation the conceptual understanding of this structure in
the setting of relational specification monads.

\km{Now we need to use them in at least one example for instance mixing state+Exn and exn+state}

\section{Embedding relational program logics}
\label{sec:embedRelProgramLogics}

\subsection{Relational Hoare Logic}
\label{sec:RHL}

As explained in the introduction, \citet{benton04relational}'s seminal
relational Hoare logic (\RHL{}) is at the origin of many works on
relational program logics (see also \autoref{sec:related}). We present
here a syntactic embedding of
\RHL{}, showing that our simple framework can host usual program logics.

Concretely, we define a translation from \While-language to monadic
programs using the $\Imp$ monad, and show that the translation of
all \citet{benton04relational}'s rules (with the exception of two
partial equivalence specific ones) are admissible in our framework
using the effect observation $\theta^\Part_{\relmarker{}}$.

The translation from direct-style imperative programs to monadic ones follows
closely \citepos{Moggi89} interpretation of call-by-value in his monadic
metalanguage. The $\Imp$ monad of \autoref{sec:effect-specific-rules}
directly interprets read and write, and
while loops are translated using the following definable combinator
\begin{lstlisting}
let while (guard: Imp bool) (body : Imp one) : Imp one =
  do_while (bindImp guard (fun b . if b then bindImp body (fun () . retImp $\true$) else retImp $\false$))
\end{lstlisting}

The proofs of admissibility for the various rules exhibit a recurrent pattern.
We first use weakening to adapt the specification obtained through the
translation to an appropriate shape for the rules of our logic. Then we use
the pure and generic monadic rules to decompose the programs on both sides.
Finally, effect-specific rules together with admissibility of the premises
finish the proof.

An easy corollary of our proof is that \citet{benton04relational}'s relational
rules are valid for our partial correctness interpretation. However, our
interpretation treats non-termination in a slightly different way from his
semantics. Indeed, our partial correctness semantics relates two programs
whenever one of them diverges, whereas his requires both program to have the
same divergence behaviour. A main difference is that our semantics is more
compositional and allows to compute a \emph{precise} specification by applying
$\theta^\Part_{\relmarker}$ to the parts and combining the results, while for
Benton's semantics this will not produce precise specifications. This makes
additional rules sound with respect to our semantics, allowing for instance to
derive that $\srelsyn{\cskip}{\cloop}{\abs{\varphi (s_1^i, s_2^i)}{\forall a_1,
    a_2, s_1^f, s_2^f . \varphi((a_1, s^f_1), (a_2, s^f_2))}}$ (equivalent to
$\srelsyn{\{\>\top\>\}~\cskip}{\cloop}{\top}$ in pre-/postcondition form),
although it is of course a choice whether one wants a semantics that validates
such rules or not. Another difference is that Benton's semantics assumes a
classical logic, in which one can "decide" termination, while our semantics
easily works in a constructive logic. We leave as future work to investigate if
Benton's semantics can be successfully expressed using a simple \emph{lax}
effect observation.

\subsection{Relational Hoare Type Theory}
\label{sec:RHTT}

\citet{NanevskiBG13} introduce Relational Hoare Type Theory (\RHTT{}) for the
specific goal of proving noninterference properties of programs.
\RHTT{} builds upon powerful but specific 
semantic objects embedded in the type theory of \coq{} to support
specifications relating two runs of a single program. We explain here how we can
reconstruct their model with a relational specification monad and an effect
observation. This connection between the two frameworks could help extending
\RHTT{} to other effects, for instance exceptions.\km{Do we want to say that the
Hoare/Dijkstra monad technology deployed there could help verifying hyperproperties?}

\paragraph{A model of state and partiality} The effects supported by \RHTT{} are
manipulation of a structured heap -- a refined version of the simple state monad
of \autoref{sec:monads} -- and partiality. In order to model these effects, a
close variant of the following monad is used
\[\M\,A = (p : heap {\to} \prop) \times (f:(r: \leq\!p) \to A {\to} heap {\to} \prop) {\times} \texttt{coherent}(f)\]
where $\leq\!p = \{r:heap {\to} \prop \mid \forall h, r\,h {\Rightarrow} p\,h
\}$ and the predicate $\texttt{coherent}$ specifies that $f$ is defined by its
value on singleton predicates consisting of only one heap. Using
predicates enables the definition of fixpoint operators, in the same
fashion as we did in our interpretation of while loops for
the \lstinline!Imp!  effect in~\autoref{sec:effect-specific-rules}.

The relational specification used by \citet{NanevskiBG13} is a variation on the simple relational
monad of stateful pre- and postconditions from \autoref{sec:simple-specs} where
the precondition only takes one input heap corresponding to the fact only one
program is considered at a time.
\[
\PPrel{}{}(A_1,A_2) = (heap \to \prop) \times{} (heap \times{} heap \to A_1 \times{} A_2 \to heap \times{} heap \to \prop)
\]

Taking the same computational monad $\M$ on both sides, that is $\M_1 = \M_2 =
\M$, we define the following simple relational effect observation $\theta : M,M
\to \PPrel{}$
\begin{align*}
  \theta (c_1, c_2) &= (\abs{h_0}{\pi_1\,c_1\,h_0 \wedge \pi_1\,c_2\,h_0},\\
  &\qquad\abs{(h_1, h_2) (a_1,a_2) (h_1',h_2')}{\pi_1\,c_1\,h_1 \wedge \pi_2 c_1 h_1 a_1 h_1'\wedge \pi_1\,c_2\,h_2 \wedge \pi_2 c_2 h_2 a_2 h_2'})
\end{align*}
\km{Conclusion?}

\section{Product programs}
\label{sec:product-programs}

\ch{Product programs as alternatives to relational logics vs
  product programs are closely related to the relational logics proofs
  (that's how I put it in related work)}

The product programs methodology is an approach to prove relational
properties that can serve as an alternative to relational program
logics~\cite{BartheDR11, BartheCK16}. In this section we show
how to understand this methodology from the point of view of our
framework.
Product programs reduce the problem of verifying relational properties on two programs
$c_1$ and $c_2$ to the problem of verifying properties on a single \emph{product
  program} $c$ capturing at the same time the behaviors of $c_1$ and $c_2$.
%
%
To prove a relational property $w$ on programs $c_1$ and $c_2$, the methodology
tells us to proceed as follows. First, we construct a product program $c$ of
$c_1$ and $c_2$. Then, by standard methods, we prove that the program $c$
satisfies the property $w$ seen as a non-relational property. Finally, from a
general argument of soundness, we can conclude that $\varphi$ must hold on $c_1$
and $c_2$.
In what follows, we show how these three
steps would be understood in our framework if we wanted to prove
$\srelsem{c_1}{c_2}{w}{\theta}$.


First of all, we need a notion of product program. In the setting of
monadic programs, we capture a product program of $c_1: \Mone{}A_1$
and $c_2 : \Mtwo{}A_2$ as a program $c:\P{}(A_1,A_2)$, where $\P{}$ is
a relative monad over $(A_1,A_2) \mapsto A_1 \times A_2$ (see \autoref{sec:RelMon}). We can think
of $c : \P{}(A_1,A_2)$ as a single computation that is computing both
a value of type $A_1$ and a value of type $A_2$ at the same time. We
expect $\P{}$ to support the effects from both $\Mone$ and $\Mtwo$,
mixing them in a controlled way. As a concrete example, we can define
products of stateful programs -- $\Mone{}A_1 = \St_{S_1}A_1$ and $\Mtwo{}A_2
= \St_{S_2}A_2$ -- inhabiting the relative monad $\P^\St(A_1,A_2) =
\St_{S_1 \times S_2}(A_1 \times A_2)$.
To complete the definition of product programs, we also need to explain
when a concrete product program $c : \P{}(A_1,A_2)$ is capturing the
behavior of $c_1 : \Mone{}A_1$ and $c_2 : \Mtwo{}A_2$. We propose to
capture this in a relation $\prodrel{c_1}{c_2}{c}$ that exhibits the
connection between pairs of computations and their potential
product programs. This relation should be closed under the
monadic construction of the effects, that is
\[
\inferrule{a_1 : A_1 \\ a_2 : A_2}{\prodrel{\ret^\Mone{}~a_1}{\ret^\Mtwo{}~a_2}{\ret^\P{}~(a_1,a_2)}}\quad
\inferrule{\prodrel{m_1}{m_2}{m_\relobj} \\ \forall a_1~a_2, \prodrel{f_1~a_1}{f_2~a_2}{f_\relobj~(a_1, a_2)}}{\prodrel{\bind^\Mone{}~m_1~f_1}{\bind^\Mtwo{}~m_2~f_2}{\bind^\P{}~m_\relobj~f_\relobj}}
\]
but also spells out how particular effects that $\P{}$ supports
correspond to the effects from $\Mone$ and $\Mtwo$.


Second, to fully reproduce the product program methodology, we need to
explain how specifications relate to product programs.
We can use simple relational specification monads~(\autoref{sec:simple-specs})
for specifying the properties on products programs. The lifting of unary
specification monads described there extends to unary effect observations,
providing an important source of examples of effect observations for product
programs.
For example, going back to the example of state, we can specify product programs in $\P{}(A_1,A_2)
= \St_{S_1 \times S_2}(A_1 \times A_2)$ with specifications provided
by the simple relational specification monad $\Wrel{\St}$, and the
effect observation $\zeta : \P{} \to \Wrel{\St}$ obtained by lifting the unary
effect observation $\theta^\St : \St \to \W^\St$ of the introduction\km{Any way
  to give a reference? currently on page 3, paragraph simple setting}, resulting in
\[
\zeta~(f : S_1 \times S_2 \to (A_1 \times A_2) \times (S_1 \times S_2)) = \abs{\varphi~\left(s_1,s_2\right)}{\varphi~\sigma(f~(s_1, s_2))}
\]
where $\sigma : (A_1 {\times} A_2) {\times} (S_1 {\times} S_2) {\to}
(A_1 {\times} S_1) {\times} (A_2 {\times} S_2)$ simply swaps the
arguments. Then, the concrete proof verifying the property $w$ in
this step consists of proving $\zeta(c) \leq w$ as usual.

Finally, the third step simply relies on proving and then applying a soundness
theorem for product programs. In the case of stateful computations, this theorem has the following form:
\begin{theorem}[Soundness of product programs for state]
If $\prodrel{c_1}{c_2}{c}$ and $\zeta(c) \leq w$, then $\srelsem{c_1}{c_2}{w}{\theta^{\St}_{\relmarker{}}}$.
\end{theorem}
In this case, the soundness theorem is proved by analyzing the relation
$\prodrel{c_1}{c_2}{c}$ and showing in each case that our choice of
$\theta^{\St}_{\relmarker{}}$ and $\zeta$ agree.

The interpretation of product programs as computations in a relative monad
accommodate well the product program methodology. In particular we expect that
algebraic presentations of these relative monads used for product programs could
shed light on the choice of primitive rules in relational program logics, in a
Curry-Howard fashion. We leave this as a stimulating future work.

\section{Related Work}
\label{sec:related}

\ch{Other stuff we could add:
\begin{itemize}
\item Proof system in LF:
For pure language constructs (1), we try to use the reasoning
principles of our ambient dependent type theory as directly as possible (in the
spirit of the use of LF in \cite{ZeilPhD}).

\item Parametricity translations for dependent type theory \cite{BernardyL11}

\item Should also relate to \cite{SwierstraB19}, who also have effect
  observations, or at least give them a cite in the text.
  
\item Anything to say about relative monads?
\end{itemize}}

Many different relational verification tools have been proposed, making
different tradeoffs, especially between automation and expressiveness.
This section surveys this prior work, starting with the techniques that are
closest related to ours.

\ch{The first 2 paragraphs below have a very repetitive phrasing, try to improve}

\paragraph{Relational program logics}
Relational program logics are very expressive and provide a formal foundation
for various tools, which have found practical applications in many domains.
\citet{benton04relational} introduced Relational Hoare Logic (RHL) as a way to
prove the correctness of various static analysis and optimizing transformations
for imperative programs.
\citet{Yang07} extended this to the relational verification of
pointer-manipulating programs.
\citepos{BartheGB09} introduced pRHL as an extension of RHL to discrete
probabilities and showed that pRHL can provide a solid foundation for
cryptographic proofs, which inspired further research in this area
\cite{PetcherM15, BartheFGSSB14, BasinLS17, Unruh19} and lead to the
creation of semi-automated tools such as EasyCrypt \cite{BartheDGKSS13}.
\citet{BartheKOB13} also applied variants of pRHL to differential privacy, which
led to the discovery of a strong connection \cite{BartheGHS17} between coupling
proofs in probability theory and relational program logic proofs, which are
in turn connected to product programs even without probabilities \cite{BartheCK16}.

\citet{CarbinKMR12} introduced a program logic for proving acceptability
properties of approximate program transformations.
\citet{NanevskiBG13} proposed Relational Hoare Type Theory (RHTT), a
verification system for proving rich information flow and access control
policies about pointer-manipulating programs in dependent type theory.
\citet{BanerjeeNN16} addressed similar problems using a relational program logic
with framing and hypotheses.
\citet{SousaD16} devised Cartesian Hoare Logic for verifying k-safety
hyperproperties and implement it in the {\sc Descartes} tool.
Finally, \citet{AguirreBGGS17} introduced Relational Higher-Order Logic (RHOL) as
a way of proving relational properties of {\em pure} programs in a simply typed
$\lambda$-calculus with inductive types and recursive definitions.
RHOL was later separately extended to two different monadic effects: cost
\cite{RadicekBG0Z18} and continuous probabilities with conditioning
\cite{SatoABGGH19}.

Each of these logics is specific to a particular combination of side-effects
that is fixed by the programming language and verification framework.
We instead introduce a general framework for defining program logics for
{\em arbitrary} monadic effects.
We show that logics such as RHL and HTT can be reconstructed within our
framework, and we expect this to be the case for many of the logics above.
It would also be interesting to investigate whether RHOL can also be extended to
arbitrary monads, but even properly representing arbitrary monads, which is
completely straightforward in dependent type theory, is not obvious in less powerful
systems such as HOL.
In this respect, \citet{Lochbihler18} recently built a library for effect
polymorphic definitions and proofs in Isabelle/HOL, based on value-monomorphic
monads and relators.

\paragraph{Type systems and static analysis tools}
Various type systems and static analysis tools have been proposed for statically
checking relational properties in a sound, automatic, but over-approximate way.
The type systems for information flow control generally trade off precision
for good automation~\cite{SabelfeldM03}.
Various specialized type systems and static analysis tools have also been
proposed for checking differential privacy \cite{Winograd-CortHR17,
  GaboardiHHNP13, BartheGAHRS15, ZhangK17, Gavazzo18, ZhangRHPR19} or doing
relational cost analysis \cite{CicekBG0H17, QuGG19}.

\paragraph{Product program constructions}
Product program constructions and self-composition are techniques
aimed at reducing the verification of k-safety hyperproperties
\cite{clarkson10hyp} to the
verification of traditional (unary) safety proprieties of a product
program that emulates the behavior of multiple input programs.
Multiple such constructions have been proposed \cite{BartheCK16} targeted for
instance at secure IFC \cite{TerauchiA05, BartheDR11, Naumann06, YasuokaT14},
program equivalence for compiler validation \cite{ZaksP08},
equivalence checking and computing semantic differences
\cite{LahiriHKR12}, program approximation \cite{HeLR18}.
\citepos{SousaD16} {\sc Descartes} tool for k-safety properties also creates k
copies of the program, but uses lockstep reasoning to improve performance by
more tightly coupling the key invariants across the program copies.
\citet{AntonopoulosGHKTW17} develop a tool \iffull called Blazer \fi that
obtains better scalability by using a new decomposition of
programs instead of using self-composition for k-safety problems.
\citet{EilersMH18} propose a modular product program construction
that permits hyperproperties in procedure specifications.
Recently, \citet{FarzanV19} propose an automated verification
technique for hypersafety properties by constructing a proof for a small
representative set of runs of the product program.

\paragraph{Logical relations and bisimulations}
Many semantic techniques have been proposed for reasoning about
relational properties such as observational equivalence, including
techniques based on binary logical relations \cite{BentonKBH09,
  Mitchell86, AhmedDR09, DreyerNRB10, DreyerAB11, DreyerNB12,
  BentonHN13, Benton0N14}, bisimulations \cite{KoutavasW06,
  SangiorgiKS11, Sumii09, LagoGL17}, and
  combinations thereof \cite{HurDNV12, HurNDBV14}.
While these powerful techniques are often not directly automated,
they can still be used for verification \cite{TimanyB19} and
for providing semantic correctness proofs for relational
program logics \cite{DreyerNRB10, DreyerAB11} and other verification
tools \cite{BentonK0N16, Gavazzo18}.

\paragraph{Other program equivalence techniques}
Beyond the ones already mentioned above, many other techniques targeted at program
equivalence have been proposed; we briefly review several recent works:
\citet{BentonKBH09} do manual proofs of correctness of
compiler optimizations using partial equivalence relations.
\citet{KunduTL09} do automatic translation validation of compiler
optimizations by checking equivalence of partially specified programs
that can represent multiple concrete programs.
\citet{GodlinS10} propose proof rules for proving the
equivalence of recursive procedures.
\citet{LucanuR15} and \citet{CiobacaLRR16} generalize this to
a set of co-inductive equivalence proof rules that are language-independent.
\citet{WangDLC18} verify equivalence between a pair of programs that operate
over databases with different schemas using bisimulation invariants over
relational algebras with updates.
Finally, automatically checking the equivalence of processes in a process
calculus is an important building block for security protocol analysis
\cite{BlanchetAF08, ChadhaCCK16}.

\paragraph{Reasoning about effectful semantics}
Relating monadic expressions is natural and very wide-spread in proof
assistants like \coq{}, Isabelle \cite{Lochbihler18}, or \fstar{}\cite{relational},
with various degrees of automation.
\citet{CasinghinoSW14, PedrotT18,BoulierPT17} extend dependent type theory with
a few selected primitive effects: partiality, exceptions, reader. The resulting
theory allows to some extent to reason directly on pairs of effectful programs,
without resorting to a monadic encoding.
In another line of work, \citet{BartheEGGKM19} proposed to encode the semantics
of imperative programs and their relational properties in an extension of
first-order logic that can be automated by Vampire.

\section{Conclusion and Future Work}
\label{sec:conclusion}

This paper introduced a principled framework for building relational program
logics.
We extended the work of \citet{dm4all} to the relational setting, and solved the
additional challenges of correlating two independent computations with a
relational specification, by we leveraging relative monads and introducing a
novel notion of relative monad morphism.
%
%
Now it's time to put this framework to the test and discover whether it can be
in part automated, whether it can be scaled to realistic relational verification
tasks, and whether it can deal with more complex effects.

In particular, it would be interesting to see whether our generic framework of
\autoref{sec:generic} can support other control effects, such as breaking out of
loops and continuations.
These are, however, challenging to accommodate, even in our unary
setting~\cite{dm4all}.
Another interesting direction is providing a more precise treatment of
nontermination.
While the simple framework of \autoref{sec:simplified} can already handle
nontermination by choosing globally between total and partial correctness with
an effect observation, the generic framework of \autoref{sec:generic} could
allow to explicitly observe whether each computation terminates or not inside
the relational specifications.
This could allow one to choose at the specification level between partial or
total correctness and between \citet{benton04relational}'s or our semantics for
RHL, or to define termination-sensitive noninterference, or that one computation
terminates whenever the other does.
Another interesting research direction, opened by the correspondence with
product programs, would be to develop techniques to select which proof
rules should be considered as primitive, using proof-theoretical tools
like focusing~\cite{ZeilPhD}, and also to investigate at the categorical level notions
of presentations of relative monads, in connection with the theory of
monads with arities~\cite{berger:hal-01296565}.
%

\ifanon\else
\begin{acks}

  We thank Alejandro Aguirre, Danel Ahman, Robert Atkey, Gilles Barthe, Shin-ya
  Katsumata, Satoshi Kura, Guido Mart\'inez, Ramkumar Ramachandra, Nikhil Swamy,
  \'Eric Tanter, and the anonymous reviewers for their helpful feedback.
  This work was, in part, supported by the
  \grantsponsor{1}{European Research Council}{https://erc.europa.eu/}
  under ERC Starting Grant SECOMP (\grantnum{1}{715753})
  and by Nomadic Labs via a grant on the
  ``Evolution, Semantics, and Engineering of the \fstar Verification System.''
  %
\end{acks}
\fi



\bibliographystyle{abbrvnaturl}
\bibliography{fstar}

\end{document}